\documentclass[aps,twocolumn,tightenlines,superscriptaddress]{revtex4-1}

\usepackage[latin1]{inputenc}
\usepackage{amsthm}
\usepackage{amssymb}
\usepackage{amsmath}
\usepackage{braket}
\usepackage{dsfont}
\usepackage{mathdots}
\usepackage{mathtools}
\usepackage{enumerate}
\usepackage{scalerel}
\usepackage{stackengine}
\usepackage[colorlinks=true, linkcolor=blue, urlcolor=blue, citecolor=red, anchorcolor=green, backref=page]{hyperref}
\usepackage{mathrsfs}
\usepackage{upgreek}

\usepackage{bm}

\usepackage{pgfplots}
\usepackage{siunitx}
\usepackage{centernot}
\usepackage{comment}
\usepackage{chngcntr}

\newtheorem{thm}{Theorem}
\newtheorem*{thm*}{Theorem}

\newtheorem*{prop*}{Proposition}
\newtheorem{lemma}[thm]{Lemma}
\newtheorem*{lemma*}{Lemma}

\newtheorem*{cor*}{Corollary}

\newtheorem*{cj*}{Conjecture}
\newtheorem{Def}[thm]{Definition}
\newtheorem*{Def*}{Definition}

\newtheorem*{question*}{Question}

\newtheorem*{problem*}{Problem}

\theoremstyle{definition}
\newtheorem{rem}[thm]{Remark}

\makeatletter
\def\thmhead@plain#1#2#3{%
  \thmname{#1}\thmnumber{\@ifnotempty{#1}{ }\@upn{#2}}%
  \thmnote{ {\the\thm@notefont#3}}}
\let\thmhead\thmhead@plain
\makeatother

\newcommand{\bb}{\begin{equation}\begin{aligned}\hspace{0pt}}
\newcommand{\bbb}{\begin{equation*}\begin{aligned}}
\newcommand{\ee}{\end{aligned}\end{equation}}
\newcommand{\eee}{\end{aligned}\end{equation*}}

\newcommand\ceil[1]{\left\lceil#1\right\rceil}
\newcommand{\texteq}[1]{\stackrel{\mathclap{\scriptsize \mbox{#1}}}{=}}

\renewcommand{\textgeq}[1]{\stackrel{\mathclap{\scriptsize \mbox{#1}}}{\geq}}
\newcommand{\ketbra}[1]{\ket{#1}\!\!\bra{#1}}
\newcommand{\ketbraa}[2]{\ket{#1}\!\!\bra{#2}}
\newcommand{\sumno}{\sum\nolimits}
\newcommand{\e}{\varepsilon}

\newcommand{\id}{\mathds{1}}
\newcommand{\R}{\mathds{R}}
\newcommand{\N}{\mathds{N}}

\newcommand{\C}{\mathds{C}}

\DeclareMathOperator{\Tr}{Tr}

\DeclareMathAlphabet{\pazocal}{OMS}{zplm}{m}{n}

\DeclareMathOperator{\supp}{supp}

\newcommand{\HH}{\pazocal{H}}
\newcommand{\T}{\pazocal{T}}

\newcommand{\CC}{\pazocal{C}}
\newcommand{\B}{\pazocal{B}}

\newcommand{\NN}{\mathcal{N}}
\newcommand{\TT}{\mathcal{T}}
\newcommand{\LL}{\mathcal{L}}

\newcommand{\Icoh}{I_{\mathrm{coh}}}

\newcommand{\SEP}{\pazocal{S}}

\newcommand{\lsmatrix}{\left(\begin{smallmatrix}}
\newcommand{\rsmatrix}{\end{smallmatrix}\right)}

\stackMath
\newcommand\xxrightarrow[2][]{\mathrel{%
  \setbox2=\hbox{\stackon{\scriptstyle#1}{\scriptstyle#2}}%
  \stackunder[2pt]{%
    \xrightarrow{\makebox[\dimexpr\wd2\relax]{$\scriptstyle#2$}}%
  }{%
   \scriptstyle#1\,%
  }%
}}

\newcommand{\tends}[2]{\xxrightarrow[\! #2 \!]{\mathrm{#1}}}

\definecolor{Blues5seq1}{RGB}{239,243,255}
\definecolor{Blues5seq2}{RGB}{189,215,231}
\definecolor{Blues5seq3}{RGB}{107,174,214}
\definecolor{Blues5seq4}{RGB}{49,130,189}
\definecolor{Blues5seq5}{RGB}{8,81,156}

\definecolor{Greens5seq1}{RGB}{237,248,233}
\definecolor{Greens5seq2}{RGB}{186,228,179}
\definecolor{Greens5seq3}{RGB}{116,196,118}
\definecolor{Greens5seq4}{RGB}{49,163,84}
\definecolor{Greens5seq5}{RGB}{0,109,44}

\definecolor{Reds5seq1}{RGB}{254,229,217}
\definecolor{Reds5seq2}{RGB}{252,174,145}
\definecolor{Reds5seq3}{RGB}{251,106,74}
\definecolor{Reds5seq4}{RGB}{222,45,38}
\definecolor{Reds5seq5}{RGB}{165,15,21}

\newcommand{\n}{a^\dag\! a}
\newcommand{\vb}[1]{\boldsymbol{\mathbf{#1}}}

\newcommand{\fakepart}[1]{
 \par\refstepcounter{part}
  \sectionmark{#1}
}

\renewcommand{\CC}{\mathscr{C}}

\usepackage{newtxtext}
\usepackage{newtxmath}

\allowdisplaybreaks

\newcommand{\deff}[1]{\textbf{\emph{#1}}}
\newcommand{\QQ}{Q_{\leftrightarrow}}

\renewcommand{\epsilon}{\varepsilon}

\newcommand{\nocontentsline}[3]{}
\newcommand{\tocless}[2]{\bgroup\let\addcontentsline=\nocontentsline#1{#2}\egroup}

\begin{document}

\title{Exact solution for the quantum and private capacities of bosonic dephasing channels}

\author{Ludovico Lami}
\email{ludovico.lami@gmail.com}
\affiliation{Institut f\"{u}r Theoretische Physik und IQST, Universit\"{a}t Ulm, Albert-Einstein-Allee 11, D-89069 Ulm, Germany}
\affiliation{QuSoft, Science Park 123, 1098 XG Amsterdam, the Netherlands}
\affiliation{Korteweg--de Vries Institute for Mathematics, University of Amsterdam, Science Park 105-107, 1098 XG Amsterdam, the Netherlands}
\affiliation{Institute for Theoretical Physics, University of Amsterdam, Science Park 904, 1098 XH Amsterdam, the Netherlands}

\author{Mark M. Wilde}
\email{wilde@cornell.edu}
\affiliation{Hearne Institute for Theoretical Physics, Department of Physics and Astronomy, and Center for Computation and Technology, Louisiana State University, Baton Rouge, Louisiana 70803, USA}
\affiliation{School of Electrical and Computer Engineering, Cornell University, Ithaca, New York 14850, USA}

\begin{abstract}
The capacities of noisy quantum channels capture the ultimate rates of information transmission across quantum communication lines, and the quantum capacity plays a key role in determining the overhead of fault-tolerant quantum computation platforms. In the case of bosonic systems, central to many applications, no closed formulas for these capacities were known for bosonic dephasing channels, a key class of non-Gaussian channels modelling, e.g., noise affecting superconducting circuits or fiber-optic communication channels. Here we provide the first exact calculation of the quantum, private, two-way assisted quantum, and secret-key agreement capacities of all bosonic dephasing channels. We prove that that they are equal to the relative entropy of the distribution underlying the channel to the uniform distribution. Our result solves a problem that has been open for over a decade, having been posed originally by~\href{https://doi.org/10.1117/12.870179}{[Jiang \& Chen, \emph{Quantum and Nonlinear Optics} 244, 2010]}.
\end{abstract}


\date{\today}
\maketitle
\let\oldaddcontentsline\addcontentsline
\renewcommand{\addcontentsline}[3]{}

One of the great promises of quantum information science is that remarkable tasks can be achieved by encoding information into quantum systems~\cite{book2000mikeandike}. In principle, algorithms executed on quantum computers can factor large integers~\cite{Shor}, simulate complex physical dynamics~\cite{childs2018toward}, solve unstructured search problems with proven speedups~\cite{grover96}, and perform linear-algebraic manipulations on large matrices encoded into quantum systems~\cite{Harrow2009,gilyen2018QSingValTransf}. Additionally, ordinary (\textquotedblleft classical\textquotedblright) information can be transmitted securely over quantum channels by means of quantum key distribution~\cite{XMZLP20}.

However, all of these possibilities are hindered in practice because all quantum systems are subject to decoherence~\cite{S05}. A very simple decoherence process takes a density operator $\rho = \sum_{n,m} \rho_{nm} \ketbraa{n}{m}$ to $\rho' = \sum_{n,m} \rho_{nm} e^{-\frac{\gamma}{2}(n-m)^2} \ketbraa{n}{m}$, where $\gamma\geq 0$ measures the extent to which the off-diagonal elements are reduced in magnitude. This process is also called dephasing, because it reduces or eliminates relative phases. Decoherence is a ubiquitous phenomenon affecting all quantum physical systems. In fact, in various platforms for quantum computation, experimentalists employ the T2 time as a phenomenological quantity that roughly measures the time that it takes for a coherent superposition to decohere to a probabilistic mixture. Dephasing noise in some cases is considered to be the dominant source of errors affecting quantum information encoded into superconducting systems~\cite{BDKS08}, as well as other platforms~\cite{Taylor2005,OLABLW08}. 
If those systems are employed to carry out quantum computation, then the errors must be amended by means of error-correcting codes, which typically causes expensive overheads in the amount of physical qubits needed. Not only does dephasing affect quantum computers, but it also affects quantum communication systems. Indeed, temperature fluctuations~\cite{W92} or Kerr non-linearities~\cite{Gordon:90} in a fiber, imprecision in the path length of a fiber~\cite{D98}, or the lack of a common phase reference between sender and receiver~\cite{BDS07} lead to decoherence as well, and this can affect quantum communication and key distribution schemes adversely.

Many of the aforementioned forms of decoherence can be unified under a single model, known as the bosonic dephasing channel (BDC)~\cite{JC10,Arqand2020}. The action of such a channel on the density operator $\rho$ of a single-mode bosonic system is given by
\bb
\NN_{p}(\rho)\coloneqq\int_{-\pi}^{\pi}d\phi\ p(\phi)\ e^{-i\n\, \phi}\rho\, e^{i\n\, \phi},\label{eq:bdc-def}
\ee
where $p$ is a probability density function on the interval $\left[-\pi,\pi\right]$ and $a^\dag a$ is the photon number operator. Since each unitary operator $e^{-i\n\,\phi}$ realizes a phase shift of the state $\rho$, the action of the channel $\NN_{p}$ is to randomize the phase of this state according to the probability density $p$. 
Representing $\rho=\sum_{n,m}\rho_{nm}\ketbraa{n}{m}$ in the photon number basis, it is a straightforward calculation to show that
\bb
\NN_{p}(\rho)=\sum_{n,m}\rho_{nm}(T_p)_{nm}\ketbraa{n}{m}\, ,
\label{eq:action-bdc-number-basis}
\ee
where $T_p$ is the infinite matrix with entries
\bb
(T_p)_{nm} \coloneqq\int_{-\pi}^{\pi}d\phi\ p(\phi)\, e^{-i\phi(n-m)}.
\label{eq:toeplitz-bos-deph}
\ee
This channel thus generalizes the simple dephasing channel considered above. Its action preserves diagonal elements of~$\rho$, but reduces the magnitude of the off-diagonal elements, a key signature of decoherence. 
As the name suggests, the BDC can be seen as a generalization to bosonic systems of the qudit dephasing channel~\cite{Devetak-Shor}.

Of primary interest is understanding the information-processing capabilities of the BDC in~\eqref{eq:bdc-def}. We can do so by means of the formalism of quantum Shannon theory~\cite{H17, W17}, in which we assume that the channel acts many times to affect multiple quantum systems. Not only does this formalism model dephasing that acts on quantum information encoded in a memory, as in superconducting systems, but also dephasing that affects communication systems. Here, a key quantity of interest is the quantum capacity $Q(\NN_{p})$ of the BDC $\NN_{p}$, which is equal to the largest rate at which quantum information can be faithfully sustained in the presence of dephasing~\cite{W17}. The quantum capacity has been traditionally studied with applications to quantum communication in mind; however, recent evidence~\cite{FMS22} indicates that it is also relevant for understanding the overhead of fault-tolerant quantum computation, i.e., the fundamental ratio of physical to logical qubits to perform quantum computation indefinitely with little chance of error. The private capacity $P(\NN_{p})$ is another operational quantity of interest~\cite{W17}, being the largest rate at which private classical information can be faithfully transmitted over many independent uses of the channel $\NN_{p}$ (Figure~\ref{protocol_Q_fig}). One can also consider both of these capacities in the scenario in which classical processing or classical communication is allowed for free between every channel use~\cite{BDSW96, TGW14IEEE}, and here we denote the respective quantities by $Q_{\leftrightarrow}(\NN_{p})$ and $P_{\!\leftrightarrow}(\NN_{p})$ (Figure~\ref{protocol_Q2_fig}). The secret-key-agreement capacity $P_{\!\leftrightarrow}(\NN_{p})$ is directly related to the rate at which quantum key distribution is possible over the channel~\cite{TGW14IEEE}, and as such, it is a fundamental limit of experimental interest.
One can also study strong converse capacities (see e.g.~\cite[Eq.~(9.122)]{H17},~\cite[Definition~9.15]{KW20book}, and~\cite{WTB16}), which sharpen the above operational interpretations by considering error probabilities between zero and one. If the usual capacity is equal to the strong converse capacity, then we say that the strong converse property holds for the channel under consideration, and the implication is that the capacity demarcates a very sharp dividing line between achievable and unachievable rates for communication. We let $Q^{\dag}(\NN_{p})$, $P^{\dag}(\NN_{p})$, $Q_{\leftrightarrow}^{\dag}(\NN_{p})$, and $P_{\!\leftrightarrow}^{\dag}(\NN_{p})$ denote the various strong converse capacities for the communication scenarios mentioned above. Understanding all of the aforementioned capacities is essential for the forthcoming quantum internet~\cite{WEH18}, which will consist of various nodes in a network exchanging quantum and private information using the principles of quantum information science.


\begin{figure}
\includegraphics[scale=0.16]{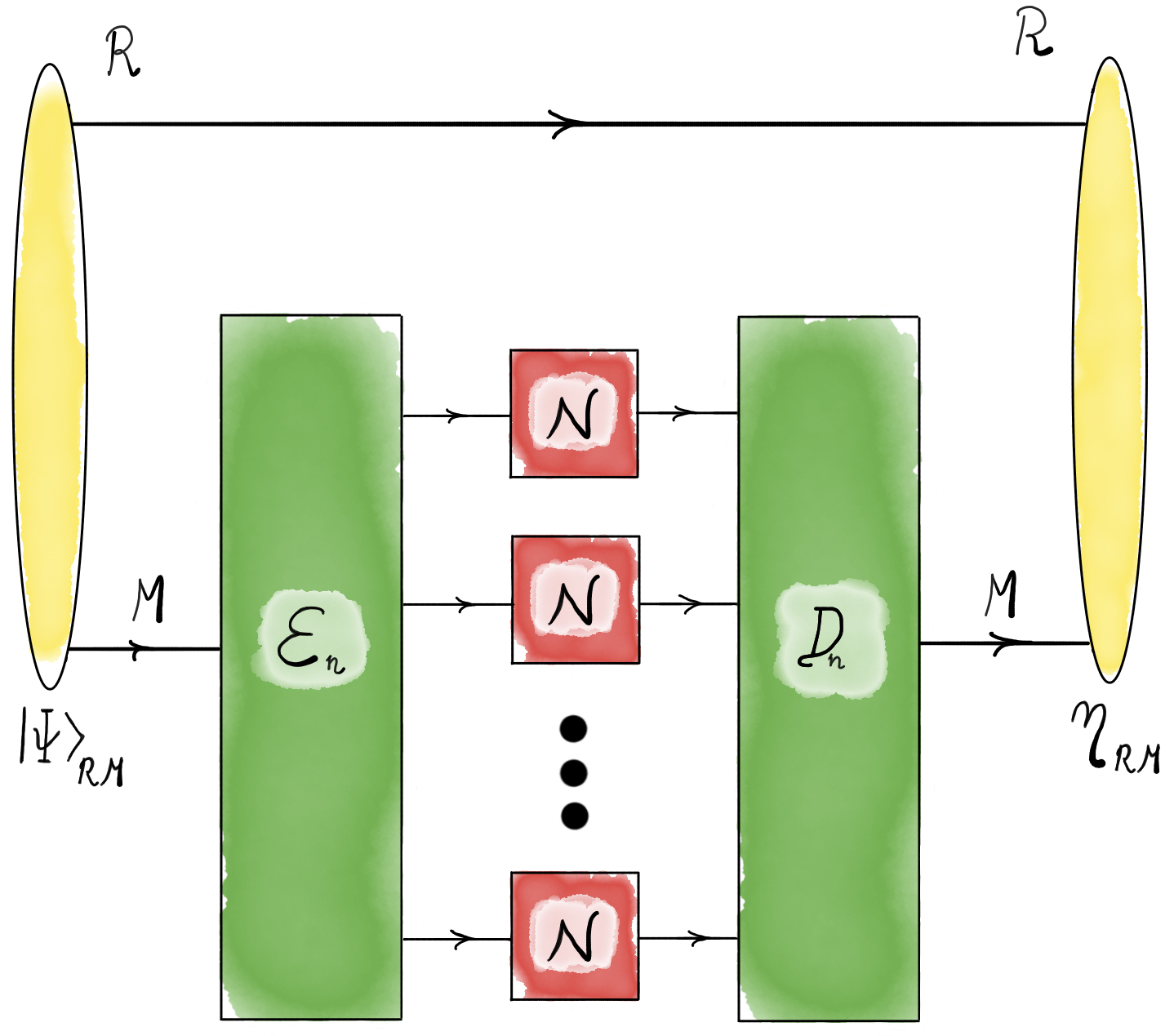}
\caption{A depiction of a quantum communication protocol that uses the channel $\NN$ a total of $n$ times to send a quantum system $M$ reliably. The initial state of the protocol is $\Psi_{RM}$, and the final state is $\eta_{RM} \coloneqq \left(\operatorname{id}_R \otimes \left(\mathcal{D}_n\circ \NN^{\otimes n} \circ \mathcal{E}_n\right)_M \right) (\Psi_{RM})$. The encoding and decoding channels ${\cal E}_n$ and ${\cal D}_n$ are operated by the sender Alice and receiver Bob, respectively.
The system $M$, initially entangled with a reference system $R$, is encoded via a suitable encoding map $\mathcal{E}_n$, transmitted via $n$ parallel uses of the channel $\NN$, and decoded at the receiving end by a decoding map $\mathcal{D}_n$. The error associated with the transmission is $\e \coloneqq \sup_{\ket{\Psi}} \left(1-\braket{\Psi_{RM}|\eta_{RM}|\Psi_{RM}}\right)$, and the number of transmitted qubits is $\log_2 |M|$, where $|M|$ is the dimension of $M$. Thus, the rate of transmitted qubits with $n$ uses of $\NN$ and error $\e$ is given by $\sup_{\mathcal{E}_n,\mathcal{D}_n} (\log_2|M|)/n \eqqcolon \frac1n Q_\e(\NN^{\otimes n})$, with the maximization being over all encoding and decoding operations achieving error at most $\e$. The quantum capacity is then obtained by taking the limit $n\to\infty$ and requiring that $\e$ vanishes in this limit, i.e.\ $Q(\NN)\coloneqq \inf_{\e\in (0,1)}\liminf_{n} \frac1n Q_\e(\NN^{\otimes n})$. The strong converse quantum capacity, instead, is constructed by allowing a nonzero error $\e$ also asymptotically, with the only requirement that it stays bounded away from its maximum value of $1$: in formula, $Q^\dag(\NN)\coloneqq \sup_{\e\in (0,1)}\limsup_{n} \frac1n Q_\e(\NN^{\otimes n})$. The private capacity $P(\NN)$ and the associated strong converse capacity $P^\dag(\NN)$ are defined analogously, with the main differences being that (a)~the transmitted message $M$ is classical, (b)~an eavesdropper Eve is granted access to all environment systems interacting with the input signals of $\NN$, and (c)~the main goal of the protocol is to transmit the message reliably in such a way that Eve does not learn about it. See~\cite{KW20book} for further expositions.}
\label{protocol_Q_fig}
\end{figure}

\begin{figure}
\includegraphics[scale=0.16]{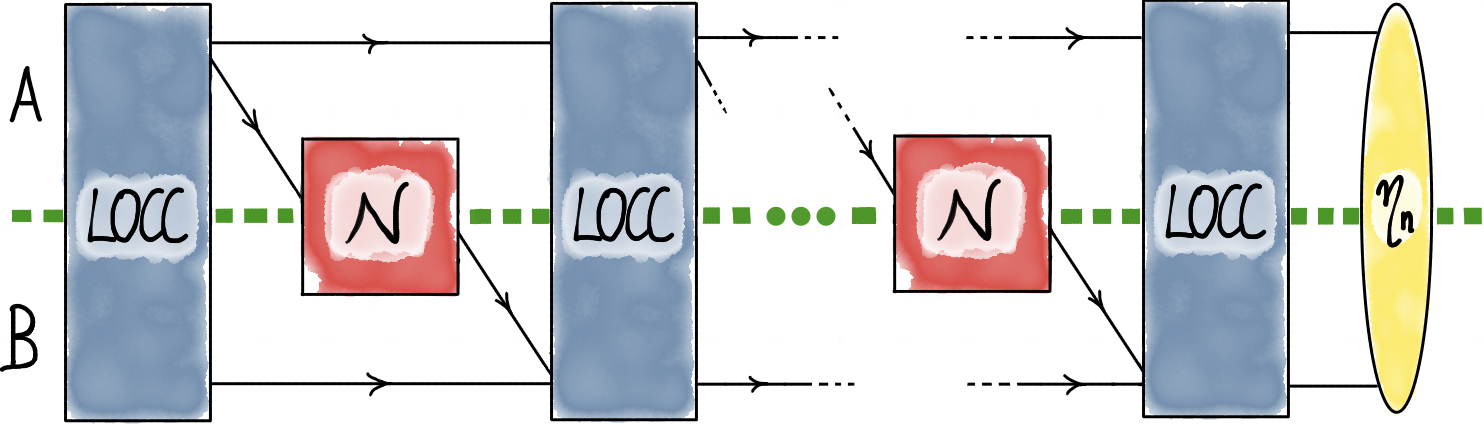}
\caption{An LOCC-assisted protocol that involves $n$ uses of the quantum channel $\NN$, assumed to connect two spatially separated laboratories, belonging to Alice and Bob. The upper arrows correspond to quantum registers of Alice and the lower arrows to quantum registers of Bob. Between each channel use and the next, Alice and Bob can implement an arbitrary protocol composed of local operations assisted by classical communication (LOCC). The final output of the protocol is a state $\eta_n$ that should resemble a maximally entangled state $\Phi_K$ of local dimension $K$. The associated error is $\e\coloneqq 1-\braket{\Phi_K|\eta_n|\Phi_K}$, and the rate of entanglement generation with $n$ uses is given by $\sup (\log_2 K)/n \eqqcolon Q_{\leftrightarrow,n,\epsilon}(\mathcal{N})$, with the maximization being over all sequences of LOCC protocols. The assisted quantum capacity of $\NN$ is then defined by taking the limit $n\to\infty$ as $Q_\leftrightarrow(\NN)\coloneqq \inf_{\e\in (0,1)} \liminf_{n} Q_{\leftrightarrow,n,\epsilon}(\mathcal{N})$, with the associated strong converse capacity being $Q^\dag_{\leftrightarrow}(\NN) \coloneqq \sup_{\epsilon \in (0,1)} \limsup_{n} Q_{\leftrightarrow,n,\epsilon}(\NN)$. The assisted private capacity $P_{\!\leftrightarrow}(\NN)$ and its strong converse capacity $P_{\!\leftrightarrow}^\dag(\NN)$ are constructed similarly, with the difference that the target state is a private state instead of a maximally entangled state.}
\label{protocol_Q2_fig}
\end{figure}

We note here that while the quantum capacity~\cite{JC10,Arqand2020} and the assisted quantum capacity~\cite{Arqand2021} of the BDC $\NN_p$ in~\eqref{eq:bdc-def} have been investigated, neither of them has been calculated so far. The determination of the quantum capacity of this channel, in particular, has been an open problem since the publication of~\cite{JC10} in~2010. The main difficulty is that $\NN_p$ is in general a non-Gaussian channel, which makes the techniques in~\cite{holwer,WPG07} inapplicable.


\medskip
\paragraph*{Results.} In this paper, we completely solve all of the aforementioned eight capacities of the BDCs, finding that they all coincide and are given by the following simple expression:
\bb
\CC(\NN_p) \coloneqq&\ \log_{2}(2\pi)\!-\!h(p) \\
=&\ Q(\NN_{p}) = P(\NN_{p}) = Q_{\leftrightarrow}(\NN_{p}) = P_{\!\leftrightarrow}(\NN_{p}) \\
=&\ Q^{\dag}\!(\NN_{p}) = P^{\dag}\!(\NN_{p}) = Q_{\leftrightarrow}^{\dag} (\NN_{p}) = P_{\!\!\leftrightarrow}^{\dag}(\NN_{p}) ,
\label{eq:main-result} 
\ee
where
\bb
h(p)\coloneqq-\int d\phi\ p(\phi)\log_{2}(p(\phi))
\ee
is the differential entropy of the probability density $p$. 
Section~III.B of the Supplementary Information contains a detailed derivation of the above result. We note here that the first expression in~\eqref{eq:main-result} can be written in terms of the relative entropy as
\bb
\log_{2}(2\pi)-h(p)=D(p\Vert u),
\ee
where $u$ is the uniform probability density on the interval $\left[-\pi,\pi\right]$, and the relative entropy is defined as
\bb
D(r\Vert s)\coloneqq\int d\phi\ r(\phi)\log_{2}\!\left(\frac{r(\phi)}{s(\phi)}\right)
\ee
for general probability densities $r$ and $s$. By invoking basic properties of relative entropy~\cite{vanErven2014}, this rewriting indicates that all of the capacities are strictly positive unless the density $p$ is uniform, which represents a complete dephasing of the channel input state.

As Eq.~\eqref{eq:main-result} indicates, there is a remarkable simplification of the capacities for BDCs. The ultimate rate of private communication over these channels is no larger than the ultimate rate for quantum communication. Furthermore, unlimited classical communication between the sender and receiver does not enhance the capacities. Finally, the strong converse property holds, meaning that the rate $D(p\Vert u)$ represents a very sharp dividing line between possible and impossible communication rates. As mentioned in the introduction, since dephasing is a prominent source of noise in both quantum communication and computation, we expect our finding to have practical relevance in both scenarios. Based on the recent findings of~\cite{FMS22}, we expect that $\left[ D(p\Vert u)\right]^{-1}$ can be related to the ultimate overhead (ratio of physical systems to logical qubits) of fault-tolerant quantum computation with superconducting systems, but further work is needed to demonstrate this definitively.

Our results can be extended to all multimode BDCs, which act simultaneously on a collection of $m$ bosonic modes with photon number operators $a_j^\dag a_j$ as
\bb
\NN_p^{(m)}(\rho) \coloneqq \int_{[-\pi,\pi]^m} \hspace{-1ex} d^m \vb{\upphi}\ p(\vb{\upphi})\, e^{-i \sum_j a^\dag_{\!j} a_{\!j}^{\phantom{\dag}} \phi_j} \rho\, e^{i\sum_j a^\dag_{\!j} a_{\!j}^{\phantom{\dag}} \phi_j} ,
\label{Np_multimode}
\ee
where $\vb{\upphi} \coloneqq (\phi_1, \ldots, \phi_m)$ and $p$ is a probability density function on the hypercube $[-\pi,\pi]^m$. The eight capacities listed in~\eqref{eq:main-result} are all equal also for the channel $\NN_p^{(m)}$, and we denote them by $\CC\big(\NN_p^{(m)}\big)$. They are given by the formula
\bb
\CC\big(\NN_p^{(m)}\big) = m\log_2(2\pi) - h(p)\, ,
\label{multimode_capacities}
\ee
where
\bb
h(p) = -\int _{[-\pi,\pi]^m}  d^m \vb{\upphi}\ p(\vb{\upphi}) \log_2 (p(\vb{\upphi}))\, ,
\ee
constituting a straightforward generalization of~\eqref{eq:main-result}. As a special case of~\eqref{Np_multimode}, when $p$ is concentrated on the line $\vb{\upphi} = (\phi,\ldots,\phi)$ and $\phi\in [-\pi,\pi]$ is uniformly distributed, one obtains the completely dephasing channel considered in~\cite{Fanizza2021squeezingenhanced, Z21}.

The most paradigmatic example of a BDC is that corresponding to a normal distribution $\widetilde{p}_\gamma(\phi) \coloneqq (2\pi\gamma)^{-1/2} e^{-\phi^2/(2\gamma)}$ of $\phi$ over the whole real line. This is the main example studied in~\cite{JC10,Arqand2020}, and it is based on a physical model discussed in those works. Here, $\gamma>0$ parametrizes the uncertainty of the rotation angle in~\eqref{eq:bdc-def}: the larger $\gamma$, the stronger the dephasing. Since values of $\phi$ that differ modulo $2\pi$ can be identified, we obtain as an effective distribution $p$ on $[-\pi,\pi]$ the \emph{wrapped normal distribution}
\bb
p_{\gamma}(\phi) \coloneqq \frac{1}{\sqrt{2\pi\gamma}} \sum_{k=-\infty}^{+\infty} e^{-\frac{1}{2\gamma} (\phi+2\pi k)^2} .
\label{wrapped_Gaussian}
\ee
The matrix $T_{p_\gamma}$ obtained by plugging this distribution into~\eqref{eq:toeplitz-bos-deph} has entries $(T_{p_\gamma})_{nm} = e^{-\frac{\gamma}{2}(n-m)^2}$, and therefore the corresponding BDC is the one discussed in the introduction. We find that
\bb
\CC(\NN_{p_\gamma}) &= \log_2 \varphi( e^{-\gamma} ) + \frac{2}{\ln 2} \sum_{k=1}^\infty \frac{(-1)^{k-1} e^{-\frac{\gamma}{2}(k^2+k)}}{k\left(1-e^{-k\gamma}\right)}\, ,
\label{capacities_wrapped_Gaussian}
\ee
where $\varphi(q) \coloneqq \prod_{k=1}^\infty \left(1-q^k\right)$ is the Euler function. See Section~IV.A of the Supplementary Information for details. In the physically relevant limit $\gamma \lesssim 1$, $p_\gamma$ and $\widetilde{p}_\gamma$ are both concentrated around $0$, and their entropies are nearly identical. In this regime
\bb
\mathscr{C}(\NN_{p_\gamma}) &\approx \frac12 \log_2\frac{2\pi}{e\gamma} \\
&\approx \left(0.604 + \frac12\log_2\frac{1}{\gamma}\right)\, \mathrm{bits}\big/\mathrm{channel\ use}\, ,
\ee
which demarcates the ultimate limitations on quantum and private communication in the presence of small dephasing noise. In the opposite case $\gamma\gg 1$ we obtain that
\bb
\mathscr{C}(\NN_{p_\gamma}) \approx \frac{e^{-\gamma}}{\ln 2}\, .
\ee
The above formula generalizes and makes quantitative the claim found in~\cite[\S~VI]{Arqand2020} that the quantum capacity of $\NN_{p_\gamma}$ vanishes exponentially for large $\gamma$. In Figure~\ref{capacities_fig}, we plot the capacity formula~\eqref{capacities_wrapped_Gaussian} as a function of $\gamma$, comparing it with the capacities $\CC(\NN_p)$ obtained for other important probability distributions $p$ on the circle.

\begin{figure}
\includegraphics[scale=0.95]{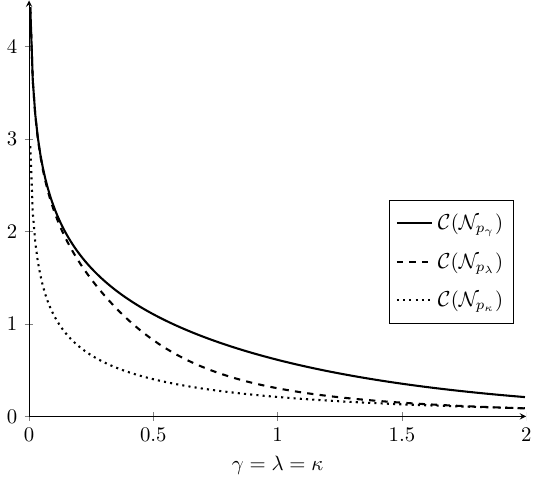}
\caption{The capacities of the BDCs associated with the wrapped normal distribution $p_\gamma$, the von Mises distribution $p_\lambda$, and the wrapped Cauchy distribution $p_\kappa$. The units of the vertical axis are qubits or private bits per channel use, and the horizontal axis features the main parameter governing the various distributions. The wrapped normal distribution is given by~\eqref{wrapped_Gaussian}, it gives a Gaussian-modulated dephasing $(T_{p_\gamma})_{nm} = e^{-\gamma (n-m)^2/2}$, and the capacity of the associated BDC $\NN_{p_\gamma}$ is given by~\eqref{capacities_wrapped_Gaussian}. The von Mises distribution $p_\lambda$ is a better analogue of the normal distribution in the case of a circle. It is given by $p_\lambda(\phi) \coloneqq \frac{e^{\cos(\phi)/\lambda}}{2\pi\, I_0(1/\lambda)}$, where $\lambda>0$ is a scale parameter analogous to $\gamma$ above, and $I_k$ is a modified Bessel function of the first kind. The obtained dephasing matrix has entries $(T_{p_\lambda})_{nm} = \tfrac{I_{|n-m|}(1/\lambda)}{I_0(1/\lambda)}$, and the capacities of the BDC $\NN_{p_\lambda}$ can be expressed as $\mathscr{C}(\NN_{p_\lambda}) = \frac{1}{\ln 2} \frac{I_1(1/\lambda)}{\lambda\, I_0(1/\lambda)} - \log_2 I_0(1/\lambda)$. Finally, the wrapped Cauchy distribution defined by $p_\kappa(\phi) \coloneqq \frac{1}{2\pi} \frac{\sinh\sqrt{\kappa}}{\cosh\sqrt{\kappa} -\cos\phi}$ corresponds to a dephasing matrix $(T_{p_\kappa})_{nm} = e^{-\sqrt{\kappa} |n-m|}$; it yields a capacity equal to $\CC(\NN_{p_\kappa}) = - \log_2\!\big(1-e^{-2\sqrt{\kappa}}\big)$.}
\label{capacities_fig}
\end{figure}

\medskip
\paragraph*{Discussion.} Our main result represents important progress for quantum information theory, solving the capacities of a physically relevant class of non-Gaussian bosonic channels. While many capacities of bosonic Gaussian channels have been solved in prior work~\cite{holwer,GGLMSY04,WPG07,Giovadd,GLMS03,PLOB,WTB16}, we are not aware of any other class of non-Gaussian channels that represent relevant models of noise in bosonic systems and whose capacity can be computed to yield a nontrivial value (neither zero nor infinite).

Our findings have non-trivial implications for the design of quantum error-correcting codes~\cite{Shor95,LB13} that encode and protect quantum information against the deleterious effects of BDCs. In particular, there is no superadditivity effect that occurs, as is the case with other quantum channels such as the depolarizing and dephrasure channels~\cite{DSS98,SS07,LLS18}. Thus, we now know that the random selection schemes of~\cite{ieee2005dev, devetak2005} are optimal designs for BDCs. It would be interesting to design quantum polar codes tailored to BDCs, as these codes are known to be capacity-achieving for certain kinds of finite-dimensional channels~\cite{RDR12,RW14}. As stated previously, another implication of our findings is that classical communication between sender and receiver does not increase the quantum and private capacities of BDCs.

Our formula can be seen as a natural generalization to bosonic systems of that given in~\cite{Devetak-Shor,TWW17,PLOB} for the quantum and private capacities of the qudit dephasing channel. However, the similarity of the final formula should not obscure the fact that the techniques used for its derivation are quite different. In particular, a key technical tool employed here is the Szeg\H{o} theorem from asymptotic linear algebra~\cite{Szego1920,SerraCapizzano2002}, in addition to a teleportation~\cite{BBCJPW93} simulation argument that is rather different from those presented previously~\cite{BDSW96,NFC09,WPG07,Alex-Master,PLOB,WTB16}.

The collapse that occurs in~\eqref{eq:main-result}, where eight different capacities are shown to coincide, also occurs for the quantum-limited bosonic amplifier channel, as a consequence of the findings of~\cite{WPG07,PLOB,Mark-energy-constrained,WTB16}. It would be interesting to determine other channels of physical interest for which this collapse occurs. It is known that this kind of collapse does not occur for the quantum erasure and pure-loss bosonic channels, because classical feedback from receiver to sender can increase the quantum and private capacities of these channels~\cite{erasure,Pirandola2009,PLOB}. Such an increase has long been known to have practical implications for the design of quantum key distribution protocols, as discussed in~\cite{Pirandola2009,PLOB}.

Going forward from here, it is of interest to address the capacities of bosonic lossy dephasing channels, in which both loss and dephasing act at the same time. Such channels are modeled as the serial concatenation $\mathcal{L}_{\eta} \circ \NN_p$, where $\mathcal{L}_{\eta}$ is a pure loss channel of transmissivity $\eta \in [0,1]$; they provide realistic noise models for communication and computation, given that both kinds of noises are relevant in these systems~\cite{LXJR22}. Our result here, combined with the main result of~\cite{WPG07} and a data-processing bottlenecking argument, leads to the following upper bound on the quantum and private capacities of the bosonic lossy dephasing channel:
\bb
Q(\mathcal{L}_{\eta} \circ \NN_p) & \leq
P(\mathcal{L}_{\eta} \circ \NN_p) \\
& \leq \min\{P(\mathcal{L}_{\eta}), P(\NN_p)\} \\
& = \min\big\{\big(\log_2(\eta / (1-\eta))\big)_+,\,D(p\Vert u)\big\}\, ,
\ee
where $x_+\coloneqq \max\{x,0\}$. By the same argument, but invoking the results of~\cite{PLOB,WTB16}, the following upper bounds hold for the quantum and private capacities assisted by classical communication:
\bb
Q_{\leftrightarrow}(\mathcal{L}_{\eta} \circ \NN_p) &\leq Q^\dag_{\leftrightarrow}(\mathcal{L}_{\eta}  \circ \NN_p),\ P_{\!\leftrightarrow}(\mathcal{L}_{\eta} \circ \NN_p) \\
& \leq P^\dag_{\!\leftrightarrow}(\mathcal{L}_{\eta}  \circ \NN_p) \\
& \leq \min\{\log_2(1 / (1-\eta)),D(p\Vert u)\}.
\ee
The same data-processing argument can be employed for BDCs composed with other common bosonic Gaussian channels in order to obtain upper bounds on the composed channels' capacities, while using known upper bounds from prior work~\cite{PLOB,WTB16,Sharma2018,Rosati2018,Noh2019,FKG21}.

It also remains open to determine the energy-constrained quantum and private capacities of BDCs, as well as their classical-communication-assisted counterparts~\cite{Arqand2020,Arqand2021}. Note that the lower bound in~\eqref{eq:max-mixed-rate} is a legitimate lower bound on the energy-constrained quantum capacity of $\NN_p$ when the mean photon number of the channel input cannot exceed $(d-1)/2$. Also, it is clear that the energy-constrained classical capacity of $\NN_p$ is equal to $g(E)\coloneqq (E+1)\log_2(E+1) - E\log_2 E$, where $E$ is the energy constraint. This identity depends essentially on the fact that Fock states can be perfectly transmitted through any BDC~\cite[\S~3.1]{Fanizza2021squeezingenhanced}. Finally, it is an open question to determine the energy-constrained entanglement-assisted classical capacity of BDCs~\cite{Holevo2004}.

In conclusion, in this work we have found an analytic expression for the quantum and private, assisted and unassisted, weak and strong converse capacities of all multimode bosonic dephasing channels, solving a problem that had been open for over a decade. BDCs are among the first non-Gaussian channels for which these capacities are calculated.

\medskip \paragraph*{Acknowledgements} --- We thank Stefano Mancini for discussions. LL was partially supported by the Alexander von Humboldt Foundation. MMW acknowledges support from the National Science Foundation under grant no.~2014010. 

\medskip
\noindent\textbf{Author contributions} --- Both authors contributed to all aspects of this manuscript and to the writing of the paper.

\medskip
\noindent\textbf{Supplementary Information} is available for this paper.

\medskip
\noindent\textbf{Competing interest} --- The authors declare no competing interests.


\bibliography{Ref,biblio}
\let\addcontentsline\oldaddcontentsline

\tocless{\section*}{Methods}

In this section, we provide a short overview of the techniques used to prove our main result~\eqref{eq:main-result}. We establish the following two inequalities:
\begin{align}
Q(\NN_{p}) &  \geq D(p\Vert u),\label{eq:lower-bnd-cap}\\
P_{\!\leftrightarrow}^{\dag}(\NN_{p}) &  \leq D(p\Vert
u).\label{eq:upper-bnd-cap}%
\end{align}
Note that~\eqref{eq:lower-bnd-cap} and~\eqref{eq:upper-bnd-cap} together imply the main result, because $Q(\NN_{p})$ is the smallest among all of the capacities listed and $P_{\!\leftrightarrow}^{\dag}(\NN_{p})$ is the largest. For a precise ordering of the various capacities, see~\cite[Eq.~(5.6)--(5.13)]{WTB16}.

To prove~\eqref{eq:lower-bnd-cap}, let us recall that the coherent information of a quantum channel is a lower bound on its quantum capacity~\cite{W17}. Specifically, the following inequality holds for a general channel $\NN$:%
\bb
Q(\NN) \geq \sup_{\rho} \left\{H(\NN(\rho))-H((\operatorname{id}\otimes\NN)(\psi^{\rho})) \right\} ,\label{eq:coh-info-low-bnd}%
\ee
where the von Neumann entropy of a state $\sigma$ is defined as $H(\sigma)\coloneqq -\Tr[\sigma\log_{2}\sigma]$, the optimization is over every state~$\rho$ that can be transmitted into the channel $\NN$, and $\psi^{\rho}$ is a purification of $\rho$ (such that one recovers $\rho$ after a partial trace). We can apply this lower bound to the BDC $\NN_{p}$. For a fixed photon number $d-1$, let us choose $\rho$ to be the maximally mixed state of dimension $d$, i.e., $\rho=\tau_{d}\coloneqq \frac{1}{d}\sum_{n=0}^{d-1}\ketbra{n}$. This state is purified by the maximally entangled state $\Phi_{d} \coloneqq \frac{1}{d}\sum_{n,m=0}^{d-1}\ketbraa{n}{m} \otimes \ketbraa{n}{m}$. To evaluate the first term in~\eqref{eq:coh-info-low-bnd}, consider from~\eqref{eq:action-bdc-number-basis} and~\eqref{eq:toeplitz-bos-deph} that the output state is maximally mixed, i.e., $\NN_p(\tau_{d})=\tau_{d}$, because the input state $\tau_{d}$ has no off-diagonal elements and the diagonal elements of the matrix $T_p$ in~\eqref{eq:toeplitz-bos-deph} are all equal to one. Thus, we find that $H(\NN_p(\tau_d))=\log_{2}d$. For the second term in~\eqref{eq:coh-info-low-bnd}, we again apply~\eqref{eq:action-bdc-number-basis} and~\eqref{eq:toeplitz-bos-deph} to determine that
\bb
\omega_{p,d} \coloneqq&\ (\operatorname{id}\otimes\NN_p)(\Phi_{d})\\
=&\ \frac{1}{d}\sum_{n,m=0}^{d-1}(T_p)_{nm}\ketbraa{n}{m} \otimes \ketbraa{n}{m}. \label{eq:choi-state-bdc}
\ee
As the entropy is invariant under the action of an isometry, and the isometry $\ket{n}\rightarrow\ket{n}\ket{n}$ takes the state 
\bb
\frac{T^{(d)}_p}{d} \coloneqq \frac{1}{d} \sum_{n,m=0}^{d-1}(T_p)_{nm}\ketbraa{n}{m}
\ee
to $\omega_{p,d}$, we find that the entropy $H(\omega_{p,d})$ reduces to%
\bb
H(\omega_{p,d})=H\Big(T^{(d)}_p\!\big/d\Big).
\ee
By a straightforward calculation, we then find that%
\bb
H\left(\NN_p(\tau_{d})\right) - H(\omega_{p,d}) &=\log_{2}d - H\Big(T^{(d)}_p\!\big/d\Big)\\
&= \frac{1}{d}\Tr \left[T^{(d)}_p \log_2 T^{(d)}_p\right]\,.
\label{eq:max-mixed-rate}
\ee
This establishes the value in~\eqref{eq:max-mixed-rate} to be an achievable rate for quantum communication over $\NN_{p}$. Since this lower bound holds for every photon number $d-1\in\mathbb{N}$, we can then take the limit $d\rightarrow\infty$ and apply the Szeg\H{o} theorem~\cite{Szego1920,SerraCapizzano2002} to conclude that the following value is also an achievable rate:
\bb
&\lim_{d\rightarrow\infty}\frac{1}{d} \Tr \left[T^{(d)}_p \log_2 T^{(d)}_p\right]\\
&\qquad =\frac{1}{2\pi}\int_{-\pi}^{\pi}d\phi\ 2\pi p(\phi)\log_{2}\left(  2\pi
p(\phi)\right)  \\
&\qquad =D(p\Vert u).
\ee
Thus, this establishes the lower bound in~\eqref{eq:lower-bnd-cap}.

To prove the upper bound in~\eqref{eq:upper-bnd-cap}, we apply a modified teleportation simulation argument. This kind of argument was introduced in~\cite[Section~V]{BDSW96}, for the specific purpose of finding upper bounds on the quantum capacity assisted by classical communication, and it has been employed in a number of works since then~\cite{NFC09,WPG07,Alex-Master,PLOB,WTB16}. Since we are interested in bounding the strong converse secret key agreement capacity $P_{\!\leftrightarrow}^{\dag}(\NN_{p})$, we apply reasoning similar to that given in~\cite{WTB16} (here see also~\cite{private,Horodecki2009}). However, there are some critical differences in our approach here.

To begin, let us again consider the state in~\eqref{eq:choi-state-bdc}. As we show in Section~III.B of the Supplementary Information, by performing the standard teleportation protocol~\cite{BBCJPW93} with the state in~\eqref{eq:choi-state-bdc} as the entangled resource state, rather than the maximally entangled state, we can simulate the action of the channel $\NN_p$ on a fixed input state, up to an error that vanishes in the limit as $d\to \infty$. This key insight demonstrates that the state in~\eqref{eq:choi-state-bdc} is approximately equivalent in a resource-theoretic sense to the channel $\NN_p$. In more detail, we can express this observation in terms of the following equality: for every state $\rho$, it holds that
\bb
\lim_{d \to \infty}\left \Vert (\operatorname{id} \otimes \NN_p)(\rho) - (\operatorname{id} \otimes \NN_{p,d})(\rho) \right\Vert_1 = 0,
\label{eq:tp-sim-error}
\ee
where $\NN_{p,d}(\sigma) \coloneqq \TT(\sigma \otimes \omega_{p,d})$ is the channel resulting from the teleportation simulation. That is, the simulating channel $\NN_{p,d}$ is realized by sending one subsystem of the maximally entangled state $\Phi_d$ through $\NN_p$, which generates $\omega_{p,d}$, and then acting on the input state $\sigma$ and the resource state $\omega_{p,d}$ with the standard teleportation protocol~$\TT$.
By invoking the main insight from~\cite{private,Horodecki2009} (as used later in~\cite{TGW14IEEE}), we next note that a protocol for secret key agreement over the channel is equivalent to one for which the goal is to distill a bipartite private state. Such a protocol involves only two parties, and thus the tools of entanglement theory come into play~\cite{private,Horodecki2009}.

Now let $\mathcal{P}_{n,\epsilon}$ denote a general, fixed protocol for secret key agreement, involving $n$ uses of the channel $\NN_p$ and achieving an error $\epsilon$ for generating a bipartite private state of rate $R_{n,\epsilon}$ (where the units of $R_{n,\epsilon}$ are secret key bits per channel use). By using the two aforementioned tools, teleportation simulation and the reduction from secret key agreement to bipartite private distillation, the protocol $\mathcal{P}_{n,\epsilon}$ can be approximately simulated by the action of a single LOCC channel on $n$ copies of the resource state $\omega_{p,d}$. Associated with this simulation are two trace norm errors $\epsilon$ and $\delta_{d}$, the first of which is the error of the original protocol $\mathcal{P}_{n,\epsilon}$ in producing the desired bipartite private state and the second of which is the error of the simulation. We then invoke~\cite[Eq.~(5.37)]{WTB16} to establish the following inequality, which, for the fixed protocol $\mathcal{P}_{n,\epsilon}$, relates the rate $R_{n,\epsilon}$ at which secret key can be distilled to the aforementioned errors and an entanglement measure called sandwiched R\'enyi relative entropy of entanglement:
\bb
R_{n,\epsilon} \leq \widetilde{E}_{R,\alpha}(\omega_{p,d}) + \frac{2\alpha}{n(\alpha-1)}\log_2\!\left(\frac{1}{1-\delta_{d} - \epsilon}\right),
\ee
where $\alpha > 1$ and the sandwiched R\'enyi relative entropy of entanglement of a general bipartite state $\rho$ is defined as~\cite{WTB16}
\bb
\widetilde{E}_{R,\alpha}(\rho) \coloneqq \inf_{\sigma \in \operatorname{SEP}} \frac{2\alpha}{\alpha-1}\log_2\!\left\Vert \rho^{1/2} \sigma^{(1-\alpha)/2\alpha}\right\Vert_{2\alpha},
\ee
with SEP denoting the set of separable (unentangled) states. By choosing the separable state to be $(\operatorname{id} \otimes \NN_p)\big(\overline{\Phi}_d\big)$, where $\overline{\Phi}_d \coloneqq \frac{1}{d}\sum_{n=0}^{d-1} \ketbra{n}\otimes \ketbra{n}$, we find that
\bb
\widetilde{E}_{R,\alpha}(\omega_{p,d}) \leq \frac{1}{\alpha-1}\log_2\frac{1}{d} \Tr\Big[\big(T_p^{(d)}\big)^\alpha\Big] .
\ee
We refer the reader to Section~III.B of the Supplementary Information for a detailed derivation. Thus, we find that the following rate upper bound holds for the secret key agreement protocol $\mathcal{P}_{n,\epsilon}$ for all $d \in \mathbb{N}$:
\begin{multline}
R_{n,\epsilon} \leq \frac{1}{\alpha-1}\log_2\frac{1}{d} \Tr\Big[\big(T_p^{(d)}\big)^\alpha\Big] \\
+ \frac{2\alpha}{n(\alpha-1)}\log_2\!\left(\frac{1}{1-\delta_{d} - \epsilon}\right),
\end{multline}
Since this bound holds for all $d\in\mathbb{N}$,
we can take the limit $d\to \infty$ and then arrive at the following upper bound:
\bb
R_{n,\epsilon} & \leq \liminf_{d\to \infty}\Bigg(\frac{1}{\alpha-1}\log_2\frac{1}{d} \Tr\Big[\big(T_p^{(d)}\big)^\alpha\Big] \label{eq:uniform-bnd-skac}\\
 & \qquad \qquad \qquad + \frac{2\alpha}{n(\alpha-1)}\log_2\!\left(\frac{1}{1-\delta_{d} - \epsilon}\right) \Bigg) \\ 
& = D_{\alpha}(p\Vert u) + \frac{2\alpha}{n(\alpha-1)}\log_2\!\left(\frac{1}{1 - \epsilon}\right).
\ee
In the above, we again applied the Szeg\H{o} theorem~\cite{Szego1920,SerraCapizzano2002} to conclude that
\bb
\lim_{d\to \infty}\frac{1}{\alpha-1}\log_2\frac{1}{d} \Tr\Big[\big(T_p^{(d)}\big)^\alpha\Big] = D_{\alpha}(p\Vert u)\, .
\ee
We also used the fact that $\lim_{d\to \infty} \delta_d = 0$, which is a consequence of~\eqref{eq:tp-sim-error}.
The bound in the last line only depends on the error $\epsilon$ of the original protocol $\mathcal{P}_{n,\epsilon}$ and the R\'enyi relative entropy
\bb
D_{\alpha}(p\Vert u)\coloneqq \frac{1}{\alpha-1}\log_2\int_{-\pi}^{\pi} d\phi\,  p(\phi)^{\alpha} u(\phi)^{1-\alpha}.
\ee
As such, it is a uniform upper bound, applying to all $n$-round secret-key-agreement protocols that generate a private state of rate $R_{n,\epsilon}$ and with error $\epsilon$. Now noting that the $n$-shot secret key agreement capacity $P_{\!\leftrightarrow}(\NN_p,n,\epsilon)$ is defined as the largest rate $R_{n,\epsilon}$ that can be achieved by using the channel $\NN_p$ a total of $n$ times along with classical communication for free, while allowing for $\epsilon$ error, it follows from the uniform bound in~\eqref{eq:uniform-bnd-skac} that
\bb
P_{\!\leftrightarrow}(\NN_p,n,\epsilon) \leq D_{\alpha}(p\Vert u) + \frac{2\alpha}{n(\alpha-1)}\log_2\!\left(\frac{1}{1 - \epsilon}\right),
\ee
holding for all $\alpha >1$. 
Remembering that the strong converse secret-key-agreement capacity is defined as
\bb
P_{\!\leftrightarrow}^\dag(\NN_p) \coloneqq \sup_{\epsilon\in(0,1)}\limsup_{n\to \infty} P_{\!\leftrightarrow}(\NN_p,n,\epsilon)
\ee
we take the limit $n\to \infty$ to find that
\bb
& P_{\!\leftrightarrow}^\dag(\NN_p) \\
& \leq \sup_{\epsilon\in(0,1)}\limsup_{n\to \infty} \left\{D_{\alpha}(p\Vert u) + \frac{2\alpha}{n(\alpha-1)}\log_2\!\left(\frac{1}{1 - \epsilon}\right) \right\} \\
& =D_{\alpha}(p\Vert u)\, .
\ee
This upper bound holds for all $\alpha>1$. Thus, we can finally take the $\alpha\to 1$ limit, and use a basic property of R\'enyi relative entropy~\cite{vanErven2014} to conclude the desired upper bound:
\bb
P_{\!\leftrightarrow}^\dag(\NN_p) \leq \lim_{\alpha \to 1}D_{\alpha}(p\Vert u)
=D(p\Vert u)\, .
\ee
This concludes the proof of the capacity formula~\eqref{eq:main-result} for the BDC. The argument required to establish its multimode generalization~\eqref{multimode_capacities} is very similar, with the only substantial technical difference being the application of the \emph{multi-index} Szeg\H{o} theorem~\cite{SerraCapizzano2002}. See Section~III.C of the Supplementary Information for details.

\medskip
\noindent\textbf{Data availability} --- No data sets were generated during this study.

\clearpage
\fakepart{Supplemental Material}

\onecolumngrid
\begin{center}
\vspace*{\baselineskip}
{\textbf{\large Supplemental Material}}\\ 
\end{center}

\renewcommand{\theequation}{S\arabic{equation}}
\renewcommand{\thethm}{S\arabic{thm}}
\renewcommand{\thefigure}{S\arabic{figure}}
\setcounter{page}{1}
\setcounter{section}{0}
\setcounter{equation}{0}
\makeatletter

\setcounter{secnumdepth}{2}


\tableofcontents

\newpage

\section{Summary of contents}

The supplementary material provides detailed arguments for the claims made in the main text. We begin with some background material on quantum states, channels, entropies, continuous-variable systems, capacities, and bosonic dephasing channels. After that, we prove our main claims about the capacities of all bosonic dephasing channels. Finally, we discuss some examples of bosonic dephasing channels of physical interest.

\section{Preliminaries, notation, and definitions}

\subsection{Quantum states and channels}

An arbitrary quantum system is mathematically represented by a separable complex Hilbert space~$\HH$. We start by reviewing a few basic concepts from the theory of operators acting on a Hilbert space $\HH$. An operator $X:\HH\to \HH$ acting on $\HH$ is  \deff{bounded} if its \deff{operator norm} $\|X\|_\infty \coloneqq \sup_{\ket{\psi}\in \HH,\, \|\ket{\psi}\|\leq 1} \left\|X\ket{\psi}\right\|$ is finite, i.e., if $\|X\|_\infty<\infty$. The Banach space of bounded operators acting on $\HH$ equipped with the norm $\|\cdot\|_\infty$ will be sometimes denoted by $\B(\HH)$. A bounded operator $X\in \B(\HH)$ is  \deff{positive semi-definite} if $\braket{\psi|X|\psi}\geq 0$ for all $\ket{\psi}\in \HH$. 

A bounded operator $T$ such that the series defining $\Tr|T| = \Tr \sqrt{T^\dag T} \eqqcolon \|T\|_1 <\infty$ converges is said to be of \deff{trace class}. Trace class operators acting on $\HH$ form another Banach space, denoted by $\T(\HH)$, once they are equipped with the trace norm $\| \cdot \|_1$. The set of positive semi-definite trace class operators forms a cone, denoted here by $\T_+(\HH)$. Since trace class operators are compact, the spectral theorem applies~\cite[Theorem~VI.16]{REED}. This means that every $T\in \T(\HH)$ can be decomposed as $T = \sum_{k=0}^\infty t_k \ketbraa{e_k}{f_k}$, where $\|T\|_1 = \sum_k |t_k|<\infty$, and $\{\ket{e_k}\}_k$, $\{\ket{f_k}\}_k$ are orthonormal bases of $\HH$.

Quantum states of a system $A$ are described by \deff{density operators}, i.e., positive semi-definite trace class operators with trace $1$, acting on $\HH_A$. The distance between two density operators $\rho,\sigma$ acting on the same Hilbert space can be measured in two different but compatible ways, either with the \deff{trace distance} $\frac12 \|\rho-\sigma\|_1$, endowed with a direct operational interpretation via the Helstrom--Holevo theorem for state discrimination~\cite{HELSTROM, Holevo1976} or with the \deff{fidelity} $F(\rho,\sigma)\coloneqq \left\|\sqrt{\rho}\sqrt{\sigma}\right\|_1^2$~\cite{Uhlmann-fidelity}. Two fundamental relations known as the Fuchs--van de Graaf inequalities establish the essential equivalence of these two distance measures. They are as follows~\cite{Fuchs1999}:
\bb
1-\sqrt{F(\rho,\sigma)} \leq \frac12 \left\|\rho-\sigma\right\|_1 \leq \sqrt{1-F(\rho,\sigma)}\, .
\label{Fuchs_vdG}
\ee

Physical transformations between states of a system $A$ and states of a system $B$ are represented mathematically as \deff{quantum channels}, i.e., completely positive trace-preserving maps $\Lambda:\T(\HH_A)\to \T(\HH_B)$~\cite{Stinespring, Choi, KRAUS}. A linear map $\Lambda:\T(\HH_A) \to \T(\HH_B)$ is 
\begin{enumerate}[(i)]
\item positive if $\Lambda\left( \T_+(\HH_A) \right) \subseteq \T_+(\HH_B))$;
\item completely positive, if $\operatorname{id}_N \otimes \Lambda : \T\big(\C^N\otimes \HH_A\big) \to \T\big(\C^N\otimes \HH_B\big)$ is a positive map for all $N\in \N$, where $\operatorname{id}_N$ represents the identity channel acting on the space of $N\times N$ complex matrices;
\item trace preserving, if $\Tr \Lambda (X) = \Tr X$ holds for all trace class $X$.
\end{enumerate}

\subsection{Entropies and relative entropies}


Let $p,q$ be two probability density functions defined on the same measurable space $\pazocal{X}$ with measure $\mu$. For $\alpha\in (0,1)\cup(1,\infty)$, define their \deff{$\boldsymbol{\alpha}$-R\'enyi divergence} by~\cite{vanErven2014}
\bb
D_\alpha(p\|q) \coloneqq \frac{1}{\alpha-1} \log_2 \int_{\pazocal{X}} d\mu(x)\ p(x)^\alpha\, q(x)^{1-\alpha}\, .
\label{Renyi}
\ee
This definition can be extended to $\alpha\in \{0,1,\infty\}$~\cite[Definition~3]{vanErven2014} by taking suitable limits. For our purposes, it suffices to consider the \deff{Kullback--Leibler} divergence~\cite{Kullback-Leibler} obtained by taking the limit $\alpha\to 1^-$ in~\eqref{Renyi}. It is defined as
\bb
D_1(p\|q) = D(p\|q) \coloneqq \int_{\pazocal{X}}d\mu(x)\ p(x)\, \log_2 \frac{p(x)}{q(x)}\, .
\label{Kullback--Leibler}
\ee
The following technical result is important for this paper.

\begin{lemma}[{\cite[Theorems~3 and~5]{vanErven2014}}] \label{technical_Renyi_lemma}
For all fixed $p,q$, the $\alpha$-R\'enyi divergence is monotonically non-decreasing in $\alpha$. Moreover, $\lim_{\alpha\to 1^-} D_\alpha(p\|q) = D(p\|q)$, and if $D_{\alpha_0}(p\|q)<\infty$ for some $\alpha_0>1$ (and therefore $D_{\alpha}(p\|q)<\infty$ for all $\alpha\in(0, \alpha_0]$) then also
\bb
\lim_{\alpha\to 1^+} D_\alpha(p\|q) = D(p\|q)\, .
\label{limit_1+_Renyi}
\ee
\end{lemma}
As a special case of the above formalism, one can define the \deff{differential R\'{e}nyi entropy} of a probability density function $p$ on $\pazocal{X}$ by setting
\bb
h_\alpha(p) \coloneqq \frac{1}{1-\alpha} \log_2 \int_{\pazocal{X}} d\mu(x)\ p(x)^\alpha ,
\ee
for all $\alpha\in (0,1)\cup(1,\infty)$. For $\alpha=1$ we obtain the standard \deff{differential entropy}, given by
\bb
h(p) \coloneqq - \int_{\pazocal{X}}d\mu(x)\ p(x) \log_2 p(x)\, ,
\label{differential_entropy}
\ee
whenever the integral is well defined.

We now consider entropies and relative entropies between quantum states. For the sake of simplicity we assume throughout this subsection that all quantum systems are finite dimensional. Indeed, in this paper we shall not consider entropies and relative entropies of infinite-dimensional states.

The most immediate way to extend~\eqref{Renyi} to the case of two quantum states $\rho,\sigma$ is to define the \deff{Petz--R\'enyi relative entropy}~\cite{PetzRenyi}
\bb
D_\alpha(\rho\|\sigma) \coloneqq \frac{1}{\alpha-1}\, \log_2 \Tr \rho^{\alpha}\sigma^{1-\alpha} ,
\label{Petz--Renyi}
\ee
where as usual $\alpha\in (0,1)\cup (1,\infty)$, and it is conventional to set $D_\alpha(\rho\|\sigma) = +\infty$ whenever $\alpha>1$ and $\supp\rho\not\subseteq \supp\sigma$, where $\supp X$ denotes the \deff{support} of $X$, i.e., the orthogonal complement of its kernel. Although~\eqref{Petz--Renyi} is a sensible definition, it is often helpful to consider an alternative quantity. The \deff{sandwiched $\boldsymbol{\alpha}$-R\'enyi relative entropy} is defined as~\cite{newRenyi, Wilde2014}
\bb
\widetilde{D}_\alpha(\rho\|\sigma) \coloneqq \frac{2\alpha}{\alpha-1} \log_2 \left\|\sigma^{\frac{1-\alpha}{2\alpha}}\rho^{\frac12} \right\|_{2\alpha} .
\label{sandwiched_Renyi}
\ee
Here, for $\beta>0$ we define the corresponding \deff{Schatten norm} of a matrix $X$ as
\bb
\|X\|_\beta \coloneqq \left(\Tr \Big[|X|^\beta\Big]\right)^{1/\beta} ,
\label{Schatten_norm}
\ee
where $|X|\coloneqq \sqrt{X^\dag X}$. As before, it is understood that $\widetilde{D}_\alpha(\rho\|\sigma) = +\infty$ when $\alpha>1$ and $\supp\rho\not\subseteq \supp\sigma$. Importantly, when $[\rho,\sigma]=0$, i.e., $\rho$ and $\sigma$ commute,~\eqref{Petz--Renyi} and~\eqref{sandwiched_Renyi} coincide, and are equal to the $\alpha$-R\'enyi divergence between the spectra of $\rho$ and $\sigma$. Namely,
\bb
[\rho,\sigma] = 0 \quad \Longrightarrow\quad \widetilde{D}_\alpha(\rho\|\sigma) = D_\alpha(\rho\|\sigma)\, .
\label{commute_Petz--Renyi}
\ee
Taking the limit as $\alpha\to 1$ of either~\eqref{Petz--Renyi} or~\eqref{sandwiched_Renyi} yields the (Umegaki) \deff{relative entropy}, given by~\cite{Umegaki1962, Lindblad1973, Hiai1991}
\bb
D(\rho\|\sigma) \coloneqq \lim_{\alpha \to 1} \widetilde{D}_\alpha(\rho\|\sigma) = \lim_{\alpha \to 1} D_\alpha(\rho\|\sigma) = \Tr\left[ \rho (\log_2 \rho - \log_2 \sigma)\right] .
\label{Umegaki}
\ee
The final quantity we need to define is the simplest of all, namely, the (von Neumann) \deff{entropy} of a density operator $\rho$:
\bb
S(\rho) \coloneqq - \Tr \left[ \rho\log_2 \rho\right] .
\label{von_Neumann}
\ee

\subsection{Continuous-variable systems}

A single-mode continuous-variable system is mathematically modelled by the Hilbert space $\HH_1\coloneqq L^2(\R)$, which comprises all square-integrable complex-valued functions over $\R$. The operators $x$ and $p \coloneqq -i\frac{d}{d x}$ satisfy the \deff{canonical commutation relation} $[x, p] = i \id$, where $\id$ denotes the identity operator (in this case, acting on $\HH_1$). Introducing the \deff{annihilation} and \deff{creation operators}
\bb
a \coloneqq \frac{x + i p}{\sqrt{2}}\, ,\qquad a^\dag \coloneqq \frac{x - i p}{\sqrt{2}}\, ,
\label{a_adag}
\ee
this can be recast in the form
\bb
[a, a^\dag] = \id \, .
\label{CCR}
\ee
Creation operators map the \deff{vacuum state} $\ket{0}$ to the \deff{Fock states}
\bb
\ket{k} \coloneqq \frac{(a^\dag)^k}{\sqrt{k!}}\, \ket{0}\, .
\label{Fock}
\ee
Fock states are eigenvectors of the \deff{photon number} operator $a^\dag a$, which satisfies
\bb
a^\dag a\,\ket{k} = k \ket{k}\, .
\label{Fock_eigenvectors}
\ee

\subsection{Unassisted capacities of quantum channels}

In this section, we briefly define the quantum and private capacities of a quantum channel. We begin with the quantum capacity. An $(|M|,\epsilon)$ code for quantum communication over the channel $\mathcal{N}_{A\to B}$ consists of an encoding channel $\mathcal{E}_{M\to A}$ and a decoding channel $\mathcal{D}_{B \to M}$ such that the channel fidelity of the coding scheme and the identity channel $\operatorname{id}_{M}$ is not smaller than $1-\epsilon$:
\bb
F( \operatorname{id}_{M},  \mathcal{D}_{B \to M}\circ \mathcal{N}_{A\to B}  \circ \mathcal{E}_{M\to A} ) \geq 1-\epsilon ,
\ee
where the channel fidelity of channels $\mathcal{N}_1$ and $\mathcal{N}_2$ is defined as \cite{GLN04}
\bb
F(\mathcal{N}_1,\mathcal{N}_2) \coloneqq \inf_{\rho} F ((\operatorname{id} \otimes \mathcal{N}_1)(\rho),(\operatorname{id} \otimes \mathcal{N}_2)(\rho)) ,
\ee
with the optimization over every bipartite state $\rho$ and the reference system allowed to be arbitrarily large. See Figure~\ref{protocol_Q_fig} for a depiction of a quantum communication protocol that uses the channel $n$ times. 
 
\begin{figure}
\includegraphics[scale=0.19]{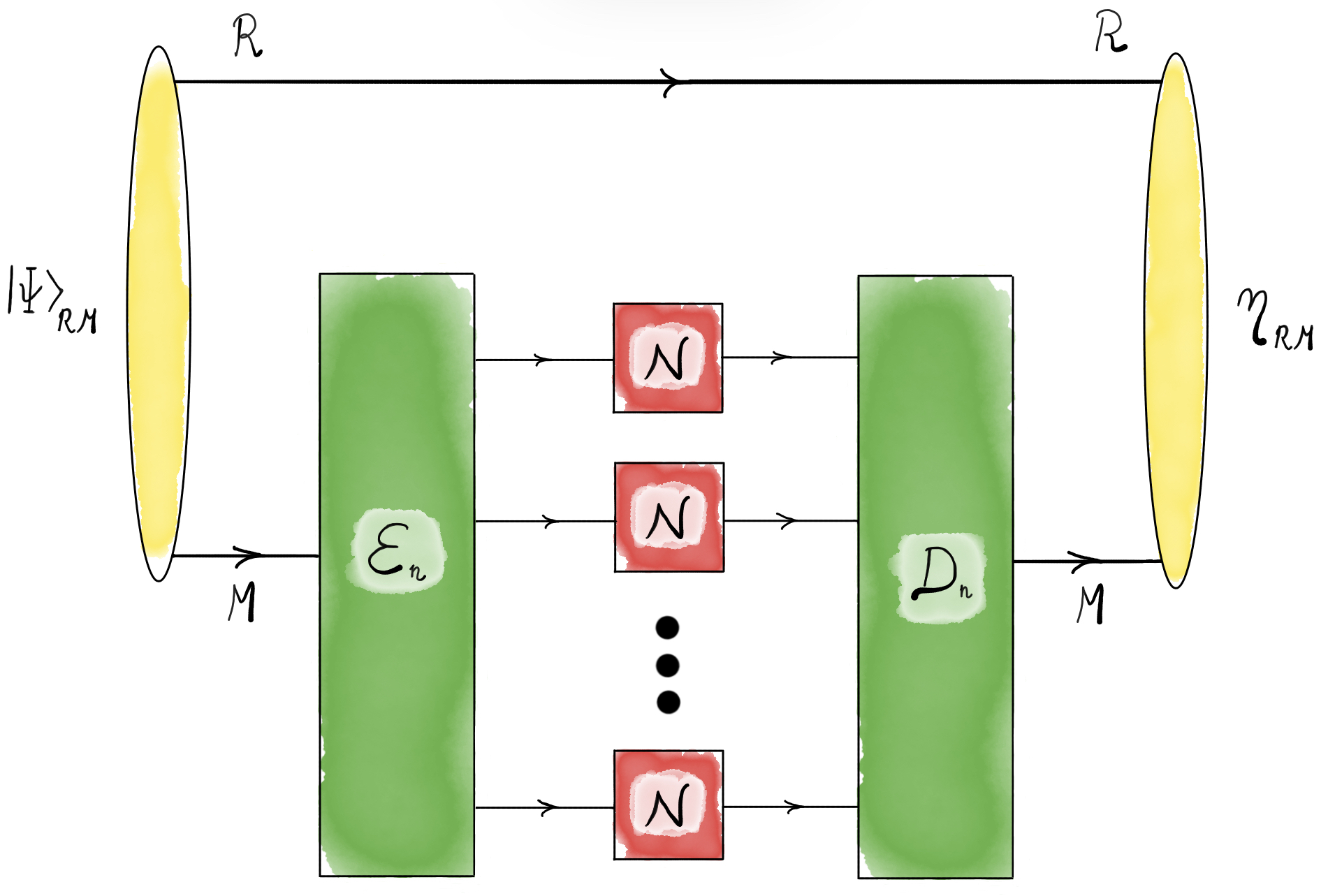}
\caption{A depiction of a quantum communication protocol that uses the channel $n$ times.}
\label{protocol_Q_fig}
\end{figure}
 
The one-shot quantum capacity $Q_{\epsilon}(\mathcal{N}_{A\to B})$ of the channel $\mathcal{N}_{A\to B}$ is defined as
\bb
Q_{\epsilon}(\mathcal{N}) \coloneqq \sup_{\mathcal{E},\mathcal{D}} \{\log_2 |M| : \exists (|M|,\epsilon) \text{ quantum communication protocol for } \mathcal{N}_{A\to B}\},
\ee
where the optimization is over every encoding channel $\mathcal{E}$ and decoding channel $\mathcal{D}$.
The (asymptotic) quantum capacity of  $\mathcal{N}_{A\to B}$ is then defined as
\bb
Q(\mathcal{N}) \coloneqq \inf_{\epsilon \in (0,1)} \liminf_{n \to \infty}\frac{1}{n}Q_{\epsilon}(\mathcal{N}^{\otimes n}),
\ee
where $\NN^{\otimes n}$ denotes $n$ copies of the channel $\NN$ used in parallel. The strong converse quantum capacity of  $\mathcal{N}_{A\to B}$ is defined as
\bb
Q^\dag(\mathcal{N}) \coloneqq \sup_{\epsilon \in (0,1)} \limsup_{n \to \infty}\frac{1}{n}Q_{\epsilon}(\mathcal{N}^{\otimes n}).
\ee
The above way of defining quantum capacity is standard, by now, in several references on quantum information theory~\cite{KW20book},~\cite[Section~VIII]{BD10}, following the same approach for defining various other capacities in classical and quantum information theory~\cite[Eqs.~(1.6)--(1.7)]{Pol10},~\cite[Section~V-A]{DMHB13},~\cite[Eq.~(1)]{TT2015},~\cite[Eq.~(10)]{Chubb2017}. There are several different ways of defining quantum capacity (see also~\cite{BS98}), but it is known that they lead to the same quantity in the asymptotic limit~\cite{temaconvariazioni}.

It is a classic result of quantum information theory that the quantum capacity is equal to the regularized coherent information of the channel~\cite{Schumacher1996,PhysRevA.54.2629,BKN98,L97,capacity2002shor,ieee2005dev}:
\bb
Q(\NN) =&\ \lim_{n\to\infty} \frac1n\, Q^{(1)}\!\left(\NN^{\otimes n}\right) = \sup_{n\in \N_+} \frac1n\, Q^{(1)}\!\left(\NN^{\otimes n}\right) , \\
Q^{(1)}\left(\NN\right) \coloneqq&\ \sup_{\ket{\Psi}_{AA'}} \Icoh(A\rangle B)_{\nu}\, ,
\label{LSD}
\ee
where
\bb
\nu_{AB} & \coloneqq \big(\operatorname{id}_A\,\otimes\, \NN_{A'\to B}\big)(\Psi_{AA'}) ,\\
\Icoh(A\rangle B)_\rho & \coloneqq S(\rho_B) - S(\rho_{AB})\, .
\label{Icoh}
\ee
This gives us a method for evaluating the quantum capacity of particular channels of interest, including the bosonic dephasing channels.

Let us now recall basic definitions related to private capacity. Let $\mathcal{U}^{\mathcal{N}}_{A\to BE}$ be an isometric channel extending the channel $\mathcal{N}_{A\to B}$~\cite{Stinespring}. An $(|M|,\epsilon)$ code for private communication over the channel $\mathcal{N}_{A\to B}$ consists of a set $\{\rho^m_A\}_m$ of encoding states and a decoder, specified as a positive operator-valued measure (POVM) $\{\Lambda^m_B\}_m$. It achieves an error $\epsilon$ if there exists a state $\sigma_E$ of the environment, such that the following inequality holds for every message $m$:
\bb
F\left( \sum_{m'} \ketbra{m'} \otimes  \Tr_B[\Lambda^{m'}_B\mathcal{U}^{\mathcal{N}}_{A\to BE}(\rho^m_A)]  , \ketbra{m} \otimes \sigma_E \right) \geq 1-\epsilon .
\ee
The one-shot private capacity $P_{\epsilon}(\mathcal{N}_{A\to B})$ of the channel $\mathcal{N}_{A\to B}$ is defined as
\bb
P_{\epsilon}(\mathcal{N}) \coloneqq \sup_{\{\rho^m_A\}_m,\{\Lambda^m_B\}_m} \{\log_2 |M| : \exists (|M|,\epsilon) \text{ private communication protocol for } \mathcal{N}_{A\to B}\},
\ee
where the optimization is over every set $\{\rho^m_A\}_m$ of encoding states and decoding POVM $\{\Lambda^m_B\}_m$.
The (asymptotic) private capacity of  $\mathcal{N}_{A\to B}$ is then defined as
\bb
P(\mathcal{N}) \coloneqq \inf_{\epsilon \in (0,1)} \liminf_{n \to \infty}\frac{1}{n}P_{\epsilon}(\mathcal{N}^{\otimes n}).
\ee
The strong converse private capacity of  $\mathcal{N}_{A\to B}$ is defined as
\bb
P^\dag(\mathcal{N}) \coloneqq \sup_{\epsilon \in (0,1)} \limsup_{n \to \infty}\frac{1}{n}P_{\epsilon}(\mathcal{N}^{\otimes n}).
\ee

The following inequalities are direct consequences of the definitions:
\bb
Q(\mathcal{N}) & \leq Q^\dag(\mathcal{N}), \\
P(\mathcal{N}) & \leq P^\dag(\mathcal{N}).
\ee
Less trivially, we also have that~\cite{ieee2005dev}
\bb
Q(\mathcal{N}) & \leq P(\mathcal{N}).
\ee

\subsection{Two-way assisted capacities of quantum channels}

In this section, we define the quantum and private capacities when the channel of interest is assisted by local operations and classical communication (LOCC). We begin with the LOCC-assisted quantum capacity.

\begin{figure}
\includegraphics[scale=0.22]{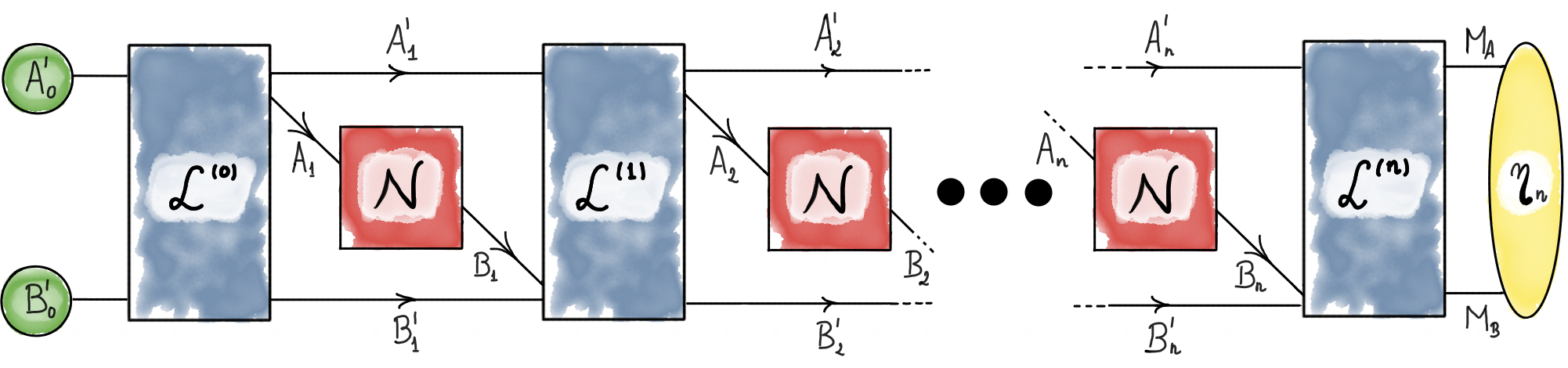}
\caption{An LOCC-assisted protocol that involves $n$ uses of the quantum channel $\NN$. Its action is described formally in Eq.~\eqref{final_state_protocol_Q2}.}
\label{protocol_Q2_fig}
\end{figure}

An $(n,|M|,\epsilon)$ protocol $\mathcal{P}$ for LOCC-assisted quantum communication consists of a separable state $\sigma_{A'_1 A_1 B'_1}$ (which is understood to be separable with respect to the bipartition $A'_1 A_1:B'_1$), the set $\big\{\mathcal{L}^{(i-1)}_{A'_{i-1}B_{i-1}B'_{i-1} \to A'_{i}A_{i}B'_{i}}\big\}_{i=2}^{n}$ of LOCC channels, and the LOCC channel $\mathcal{L}^{(n)}_{A'_{n}B_{n}B'_{n}\to M_A M_B}$. (See~\cite{LOCC} for the definition of an LOCC channel.) We can also imagine that the state $\sigma_{A'_1 A_1 B'_1}$ is produced by an LOCC preprocessing channel $\mathcal{L}^{(0)}_{A'_0B'_0\to A'_1A_1B'_1}$, with the $A'_0B'_0$ system initialised in a product state. The final state of the protocol is
\bb
\eta_{M_AM_B} \coloneqq \Big(\mathcal{L}^{(n)}_{A'_{n}B_{n}B'_{n}\to M_A M_B} \circ \mathcal{N}_{A_n \to B_n} &\circ \mathcal{L}^{(n-1)}_{A'_{n-1}B_{n-1}B'_{n-1} \to A'_{n}A_{n}B'_{n}} \circ \cdots\\
&\circ  \mathcal{L}^{(1)}_{A'_{1}B_{1}B'_{1} \to A'_{2}A_{2}B'_{2}}\circ  \mathcal{N}_{A_1 \to B_1}\Big)\big(\sigma_{A'_1 A_1 B'_1}\big)\, ,
\label{final_state_protocol_Q2}
\ee
satisfying
\bb
F(\eta_{M_A M_B}, \Phi_{M_A M_B}) \geq 1- \epsilon,
\ee
where $\Phi_{M_A M_B}$ is a maximally entangled state of Schmidt rank $|M|$. Such a protocol is depicted in Figure~\ref{protocol_Q2_fig}. We note here that it suffices for such a protocol to generate the maximally entangled state $\Phi_{M_A M_B}$ because entanglement and quantum communication are equivalent communication resources when classical communication is freely available, due to the teleportation protocol~\cite{teleportation}. 

The $n$-shot LOCC-assisted quantum capacity of the channel $\mathcal{N}_{A\to B}$ is defined as
\bb
Q_{\leftrightarrow,n,\epsilon}(\mathcal{N}) \coloneqq
\sup_{\mathcal{P}} \left\{\frac{1}{n} \log_2 |M| : \exists (n,|M|,\epsilon) \text{ LOCC-assisted q.~comm.~protocol } \mathcal{P} \text{ for } \mathcal{N}_{A\to B} \right \},
\ee
where the optimization is over every LOCC-assisted quantum communication protocol $\mathcal{P}$.
The (asymptotic) LOCC-assisted quantum capacity of $\mathcal{N}_{A\to B}$ is then defined as
\bb
Q_{\leftrightarrow}(\mathcal{N}) \coloneqq \inf_{\epsilon \in (0,1)} \liminf_{n \to \infty}Q_{\leftrightarrow,n,\epsilon}(\mathcal{N}).
\ee
The strong converse LOCC-assisted quantum capacity of $\mathcal{N}_{A\to B}$ is defined as
\bb
Q^\dag_{\leftrightarrow}(\mathcal{N}) \coloneqq \sup_{\epsilon \in (0,1)} \limsup_{n \to \infty}Q_{\leftrightarrow,n,\epsilon}(\mathcal{N}).
\ee

An $(n,|M|,\epsilon)$ protocol $\mathcal{K}$ for secret key agreement over a quantum channel is defined essentially the same as an LOCC-assisted protocol for quantum communication, except that the target final state of the protocol is more general. That is, the final step of the protocol is an LOCC channel $\mathcal{L}^{(n)}_{A'_{n}B_{n}B'_{n}\to M_A M_B S_A S_B}$, where $S_A$ and $S_B$ are extra systems of the sender Alice and the receiver Bob. Let us then denote the final state of the protocol by $\eta_{M_A M_B S_A S_B}$. Such a protocol satisfies
\bb
F(\eta_{M_A M_B S_A S_B},\gamma_{M_A M_B S_A S_B}) \geq 1-\epsilon,
\ee
where $\gamma_{M_A M_B S_A S_B} $ is a private state of dimension $|M|$~\cite{private,Horodecki2009}, having the form
\bb
\gamma_{M_A M_B S_A S_B} \coloneqq U_{M_A M_B S_A S_B} (\Phi_{M_A M_B} \otimes \theta_{S_A S_B} )U_{M_A M_B S_A S_B}^\dag.
\ee
In the above, $U_{M_A M_B S_A S_B}$ is a twisting unitary of the form
\bb
U_{M_A M_B S_A S_B} = \sum_{i,j} \ketbra{i}_{M_A} \otimes \ketbra{j}_{M_B} \otimes U^{i,j}_{S_A S_B},
\ee
with each $U^{i,j}_{S_A S_B}$ a unitary. Also, $\Phi_{M_A M_B}$ is a maximally entangled state of Schmidt rank $|M|$ and $\theta_{S_A S_B}$ is an arbitrary state. The fact that such a protocol is equivalent to the more familiar notion of secret key agreement, involving three parties generating a tripartite secret key state of the form $\frac{1}{|M|} \sum_{m=0}^{|M|-1}\ketbra{m}_{M_A} \otimes \ketbra{m}_{M_B} \otimes \sigma_E$, is the main contribution of~\cite{private,Horodecki2009} (see~\cite{KW20book} for another presentation).

The $n$-shot secret-key-agreement capacity of the channel $\mathcal{N}_{A\to B}$ is defined as
\bb
P_{\!\leftrightarrow,n,\epsilon}(\mathcal{N}) \coloneqq
\sup_{\mathcal{K}} \left\{\frac{1}{n} \log_2 |M| : \exists (n,|M|,\epsilon) \text{ secret-key-agreement protocol } \mathcal{K} \text{ for } \mathcal{N}_{A\to B} \right \},
\ee
where the optimization is over every secret key agreement protocol $\mathcal{K}$.
The (asymptotic) secret key agreement capacity of  $\mathcal{N}_{A\to B}$ is then defined as
\bb
P_{\!\leftrightarrow}(\mathcal{N}) \coloneqq \inf_{\epsilon \in (0,1)} \liminf_{n \to \infty}P_{\!\leftrightarrow,n,\epsilon}(\mathcal{N}).
\ee
The strong converse secret key agreement capacity of $\mathcal{N}_{A\to B}$ is defined as
\bb
P^\dag_{\!\leftrightarrow}(\mathcal{N}) \coloneqq \sup_{\epsilon \in (0,1)} \limsup_{n \to \infty}P_{\!\leftrightarrow,n,\epsilon}(\mathcal{N}).
\label{strong_converse_secret_key_agreement_capacity}
\ee

The following inequalities are direct consequences of the definitions:
\bb
Q_{\leftrightarrow}(\mathcal{N}) & \leq Q_{\leftrightarrow}^\dag(\mathcal{N}) \\
P_{\!\leftrightarrow}(\mathcal{N}) & \leq P_{\!\leftrightarrow}^\dag(\mathcal{N}).
\ee
Due to the fact that a more general target state is allowed in secret key agreement, the following inequalities hold
\bb
Q_{\leftrightarrow}(\mathcal{N}) & \leq P_{\!\leftrightarrow}(\mathcal{N})\\
Q^{\dag}_{\leftrightarrow}(\mathcal{N}) & \leq P^{\dag}_{\leftrightarrow}(\mathcal{N}).
\ee

Finally, due to the fact that classical communication can only enhance capacities, the following inequalities hold:
\bb
Q(\mathcal{N}) & \leq Q_{\leftrightarrow}(\mathcal{N})\\
Q^{\dag}(\mathcal{N}) & \leq Q^{\dag}_{\leftrightarrow}(\mathcal{N})\\
P(\mathcal{N}) & \leq P_{\!\leftrightarrow}(\mathcal{N})\\
P^{\dag}(\mathcal{N}) & \leq P^{\dag}_{\leftrightarrow}(\mathcal{N}).
\ee
Thus, to establish the collapse of all of the capacities discussed in this section and the previous one, for the case of bosonic dephasing channels, it suffices to prove the lower bound on $Q(\mathcal{N})$ and the upper bound on $P^{\dag}_{\leftrightarrow}(\mathcal{N})$.

\subsection{Teleportation simulation}

The $d$-dimensional quantum teleportation protocol~\cite{teleportation} takes as input a $d$-dimensional quantum state $\rho_{A'}$ of a system $A'$, a $d$-dimensional maximally entangled state
\bb
\Phi_d^{AB} \coloneqq \ketbra{\Phi_d}_{AB}\, ,\qquad \ket{\Phi_d}_{AB} \coloneqq \frac{1}{\sqrt{d}} \sum_{k=0}^{d-1} \ket{k}_A\ket{k}_B\, ,
\label{Phi_d}
\ee
and by using only local operations and one-way classical communication from Alice to Bob reproduces the exact same state $\rho$ on the system $B$. To define it rigorously, for $x,z\in \{0,\ldots,d-1\}$ let us introduce the unitary matrices
\bb
X(x)\,\coloneqq\, \sum_{k=0}^{d-1} \ket{k \oplus x}\!\!\bra{k}\, ,\qquad Z(z)\, \coloneqq\, \sum_{k=0}^{d-1} e^{\frac{2\pi i}{d} zk} \ketbra{k}\, ,\qquad U(x,z)\coloneqq X(x) Z(z)\, ,
\label{HW}
\ee
where $\oplus$ denotes sum modulo $d$. Then the teleportation channel $\T^{(d)}_{A'AB\to B}$ is given by
\bb
\T^{(d)}_{A'AB\to B}(X_{A'AB}) \coloneqq \sum_{x,z=0}^{d-1} U(x,z)_{B} \Tr_{AA'}\!\left[ X_{A'AB}\, U(x,z)_{A'} \Phi_d^{AA'} U(x,z)_{A'}^\dag \right] U(x,z)_{B}^{\dag}\, .
\label{telep}
\ee
The effectiveness of the standard quantum teleportation protocol is expressed by the identity
\bb
\T^{(d)}_{A'AB\to B}\!\left(\rho_{A'} \otimes \Phi_d^{AB}\right)=\rho_{B}\, ,
\ee
meaning that the same operator $\rho$ is written in the registers $A'$ and $B$ on the left-hand and on the right-hand side, respectively.

Some channels can be simulated by the action of the standard teleportation protocol on their Choi states~\cite{BDSW96}, in the sense that
\bb
\T^{(d)}_{A'AB\to B}\!\left(\rho_{A'} \otimes \Phi_{\mathcal{N}}^{AB}\right)=\mathcal{N}(\rho_{A'})\, ,
\ee
where $\Phi_{\mathcal{N}}^{AB}$ is the Choi state of the channel $\mathcal{N}$.
For example, this is the case for all Pauli channels. More generally, other channels can be simulated approximately by the action of the teleportation protocol on their Choi states. This concept was introduced in~\cite{BDSW96} for the explicit purpose of obtaining upper bounds on the LOCC-assisted quantum capacity of a channel in terms of an entanglement measure evaluated on the Choi state. The idea was rediscovered in~\cite{Mul12} for the same purpose, and more recently the same idea was used to bound the secret-key-agreement capacity~\cite{PLOB} and the strong converse secret-key-agreement capacity~\cite{WTB16}. Here we make use of this concept in order to obtain upper bounds on the strong converse secret key agreement capacity of all bosonic dephasing channels. As discussed earlier, it suffices to consider establishing an upper bound on this latter capacity because it is the largest among all the capacities that we consider in this paper.

\subsection{Bosonic dephasing channel}

\begin{Def}
Let $p$ be a probability density function on the interval $[-\pi,\pi]$. The associated \deff{bosonic dephasing channel} is the quantum channel $\NN_p:\T(\HH_1)\to \T(\HH_1)$ acting on a single-mode system and given by
\bb
\NN_p (\rho) \coloneqq \int_{-\pi}^{\pi} d\phi\ p(\phi)\, e^{-i \n\, \phi} \rho\, e^{i \n\, \phi}\, ,
\label{Np}
\ee
where $a^\dag a$ is the photon number operator.
\end{Def}

The action of the bosonic dephasing channel can be easily described by representing the input operator in the Fock basis. By means of this representation the Hilbert space of a single-mode system, $\HH_1$, becomes equivalent to that of square-summable complex-valued sequences, denoted $\ell^2(\N)$. Operators on $\HH_1$ are represented by \deff{infinite matrices}, i.e., operators $S:\ell^2(\N)\to \ell^2(\N)$. Given two such operators $S,T$, which we formally write $S = \sum_{h,k} S_{hk} \ketbraa{h}{k}$ and $T = \sum_{h,k} T_{hk} \ketbraa{h}{k}$, their \deff{Hadamard product} is defined by
\bb
S\circ T \coloneqq \sum_{h,k} S_{hk} T_{hk} \ketbraa{h}{k}\, .
\label{Hadamard_product}
\ee
One of the fundamental facts concerning the Hadamard product is the \emph{Schur product theorem}~\cite[Theorem~7.5.3]{HJ1}: it states that if $S\geq 0$ and $T\geq 0$ are positive semi-definite, then also $S\circ T\geq 0$ is such. The theorem is usually stated for matrices, but is is immediately generalisable to the operator case as a consequence of the remark below.

\begin{rem} \label{positive_semidefiniteness_truncation_rem}
Let $T: \ell^2(\N)\to \ell^2(\N)$ be an infinite matrix. Then $T\geq 0$ if and only if $T^{(d)}\geq 0$ for all $d\in \N_+$, where $T^{(d)}$ is the $d\times d$ top left corner of $T$. This follows from the fact that the linear span of the basis vectors $\ket{k}$, $k\in\N$, is dense in $\ell^2(\N)$.
\end{rem}

Given an infinite matrix $T$ which represents a bounded operator $T:\ell^2(\N) \to \ell^2(\N)$, we can define the associated \deff{Hadamard channel} as
\bb
\begin{array}{ccccc} \LL_T & : & \T\left(\ell^2(\N)\right) & \longrightarrow & \T\left(\ell^2(\N)\right) \\[1ex]
&& S & \longmapsto & S\circ T\, . \end{array}
\label{Hadamard_channel}
\ee
When restricted to finite-dimensional systems, Hadamard channels are examples of so-called `partially coherent direct sum channels'~\cite{Chessa2021}. The following is then easily established.

\begin{lemma}
Let $T:\ell^2(\N) \to \ell^2(\N)$ be a bounded operator represented by an infinite matrix. Then the Hadamard channel $\LL_T$ defined by~\eqref{Hadamard_channel} is a completely positive and trace preserving map, i.e., a quantum channel, if and only if
\begin{enumerate}[(i)]
\item $T\geq 0$ as an operator; and
\item $T_{kk} = 1$ for all $k\in \N$.
\end{enumerate}
\end{lemma}

\begin{proof}
The two conditions are clearly necessary. In fact, if $T_{kk}\neq 1$ for some $k\in \N_+$, then $\Tr[T\circ \ketbra{k}] = T_{kk} \neq 1 = \Tr \ketbra{k}$; i.e., $\LL_T$ is not trace preserving. Also, if $T \ngeq 0$ then by Remark~\ref{positive_semidefiniteness_truncation_rem} there exists $d\in \N_+$ and some $\ket{\psi}\in \C^d$ such that $\braket{\psi|T^{(d)}|\psi} < 0$. Rewriting $\braket{\psi|T^{(d)}|\psi} = \sum_{h,k=0}^{d-1} \psi_h^* \psi_k T_{hk} = d \braket{+| (T\circ \psi)|+}$, where $\psi\coloneqq \ketbra{\psi}$ and $\ket{+}\coloneqq \frac{1}{\sqrt{d}} \sum_{k=0}^{d-1} \ket{k}$, shows that  in this case $\LL_T$ would not even be positive, let alone completely positive.

Vice versa, conditions~(i)--(ii) are sufficient. In fact, on the one hand, by~(ii), for an arbitrary $X$, we have that $\Tr \LL_T(X) = \sum_k T_{kk} X_{kk} = \Tr X$, i.e., $\LL_T$ is trace preserving. On the other hand, if $T\geq 0$ then for all $d\in \N_+$ and for all positive semi-definite bipartite operators $X \geq 0$ acting on $\C^d \otimes \ell^2(\N)$ we have that $\left(I \otimes \LL_T\right) (X) = d \left(\ketbra{+} \otimes T\right) \circ X \geq 0$, where $\ket{+}$ is defined above, and the last inequality follows by the Schur product theorem. Since $d$ is arbitrary, this proves that $\LL_T$ is completely positive.
\end{proof}

The theory of Hadamard channels we just sketched out is relevant here due to the following simple observation.

\begin{lemma}
When both the input and the output density operators are represented in the Fock basis, the bosonic dephasing channel $\NN_p$ acts as the Hadamard channel
\begin{align}
\NN_p(\rho) =&\ \rho\circ T_p\, , \label{Np_action_Tp} \\
(T_p)_{hk} \coloneqq&\ \int_{-\pi}^{\pi} d\phi\ p(\phi)\, e^{-i\phi (h-k)} \, . \label{Tp}
\end{align}
\end{lemma}

\begin{proof}
Due to~\eqref{Fock_eigenvectors}, we have that $\NN_p(\rho) = \sum_{h,k} \rho_{hk} \int_{-\pi}^{\pi}d\phi\ p(\phi)\, e^{-i\phi (h-k)} \ketbraa{h}{k} = \rho\circ T_p$.
\end{proof}

\section{Capacities of bosonic dephasing channels}

\subsection{Infinite Toeplitz matrices and theorems of Szeg\H{o} and Avram--Parter type}

Observe that the expression for $(T_p)_{hk}$ only depends on the difference $h-k$. Matrices with this property are named after the mathematician Otto Toeplitz. Formally, a \deff{Toeplitz matrix} of size $d\in \N_+$ is a matrix of the form
\bb
T = \begin{pmatrix} a_0 & a_{-1} & a_{-2} & \ldots & a_{-d+1} \\
a_1 & a_0 & a_{-1} & \ldots & a_{-d+2} \\
a_2 & a_1 & a_0 & \ldots & a_{-d+3} \\
\vdots & \vdots & \vdots & \ddots & \vdots \\
a_{d-1} & a_{d-2} & a_{d-3} & \ldots & a_0 \end{pmatrix} ,
\ee
where $a_0,\ldots, a_{d-1}\in \C$. Alternatively, it can be defined to have entries
\bb
T_{hk} = a_{h-k}\, .
\label{Toeplitz_entries}
\ee
This definition can be formally extended to the case of \deff{infinite Toeplitz matrices}, simply by letting $h,k\in \N$ run over all non-negative integers. Note that the top left corners $T^{(d)}\coloneqq \sum_{h,k=0}^{d-1} T_{hk} \ketbraa{h}{k}$ of an infinite Toeplitz matrix are Toeplitz matrices themselves.

In applications one often encounters the case in which the numbers $a_k$ are the Fourier coefficients of an absolutely integrable function $a:[-\pi,\pi]\to \C$, i.e.
\bb
a_k = \int_{-\pi}^{+\pi} \frac{d\phi}{2\pi}\ a(\phi)\, e^{-ik\phi} .
\label{a_k_Fourier}
\ee
In this paper, we will consider mainly non-negative functions $a:[-\pi,\pi]\to \R_+$.

A result due to Szeg\H{o}~\cite{Szego1920, GRENADER} states that the spectrum of the $d\times d$ top left corners $T^{(d)}$ of an infinite Toeplitz matrix converges to the generating function $a:[-\pi,\pi]\to \R$ (for now assumed to be real-valued), in the sense that
\bb
\lim_{d\to\infty} \frac1d \Tr F\big( T^{(d)}\big) = \lim_{d\to\infty} \frac1d \sum_{j=1}^d F\big(\lambda_j\big(T^{(d)}\big)\big) = \int_{-\pi}^{\pi} \frac{d\phi}{2\pi}\ F\big(a(\phi)\big)
\label{Szego}
\ee
\emph{whenever $a$ and $F:\R\to \R$ are sufficiently well behaved.} Here, $\lambda_j\big(T^{(d)}\big)$ denotes the $j^\text{th}$ eigenvalue of the matrix $T^{(d)}$. The scope and extension of Szeg\H{o}'s result has been expanded over the years by relaxing the conditions to be imposed on $a$ and $F$ so that~\eqref{Szego} holds. At the same time, an analogous class of results, initially conceived by Parter~\cite{Parter1986} and Avram~\cite{Avram1988}, has been developed to deal with the case of complex-valued generating functions $a:[-\pi,\pi]\to \C$. Results of the Avram--Parter type generalize~\eqref{Szego} by stating that
\bb
\lim_{d\to\infty} \frac1d \Tr F\Big( \big|T^{(d)}\big|\Big) = \lim_{d\to\infty} \frac1d \sum_{j=1}^d F\big(s_j \big(T^{(d)}\big)\big) = \int_{-\pi}^{\pi} \frac{d\phi}{2\pi}\ F\Big(\big| a(\phi)\big|\Big)\, ,
\label{Avram-Parter}
\ee
where $s_j\big(T^{(d)}\big)$ is now the $j^\text{th}$ \emph{singular value} of the matrix $T^{(d)}$. Both Szeg\H{o}'s and Avram--Parter's result have been generalized in successive steps, by Zamarashkin and Tyrtyshnikov~\cite{Zamarashkin1997}, Tilli~\cite{Tilli1998}, Serra-Capizzano~\cite{SerraCapizzano2002}, B\"{o}ttcher, Grudsky, and Maksimenko~\cite{Boettcher2008}, and others. For a detailed account of these developments, we refer the reader to the textbooks~\cite{GRUDSKY, BOETTCHER, GARONI} and especially to the lecture notes by Grudsky~\cite{Grudsky-lectures}. Here we will just need the following lemma, extracted from the work of Serra-Capizzano.

\begin{lemma}[(Serra-Capizzano~\cite{SerraCapizzano2002})] \label{Serra-Capizzano_lemma}
If $a:[-\pi,\pi]\to \R_+$ is such that
\bb
\int_{-\pi}^{+\pi} \frac{d\phi}{2\pi}\ a(\phi)^\alpha < \infty
\ee
for some $\alpha \geq 1$, and moreover $F:\R_+ \to \R$ is continuous and satisfies
\bb
F(x) = O(x^\alpha) \qquad (x\to \infty)\, ,
\ee
then~\eqref{Szego} holds.
\end{lemma}

\begin{proof}
The original result by Serra-Capizzano~\cite[Theorem~2]{SerraCapizzano2002} states that the identity~\eqref{Avram-Parter} involving singular values holds. However, under the stronger hypotheses that we are making here it can be seen that~\eqref{Szego} and~\eqref{Avram-Parter} are actually equivalent. Since $a$ takes on values in $\R_+$, we only need to check that for each $d$ the singular values and the eigenvalues of $T^{(d)}$ coincide. To this end, it suffices to note that $T^{(d)}$ is a positive semi-definite operator, simply because
\bb
\sum_{h,k=0}^{d-1} \psi_h^* \psi_k T_{hk} &= \sum_{h,k=0}^{d-1} \psi_h^* \psi_k \int_{-\pi}^{+\pi} \frac{d\phi}{2\pi}\ a(\phi)\, e^{-i(h-k)\phi} \\
&= \int_{-\pi}^{+\pi} \frac{d\phi}{2\pi}\ a(\phi) \sum_{h,k=0}^{d-1} \psi_h^* \psi_k\, e^{-i(h-k)\phi} \\
&= \int_{-\pi}^{+\pi} \frac{d\phi}{2\pi}\ a(\phi) \left|\sumno_{k=0}^{d-1} \psi_k\, e^{ik\phi} \right|^2 \\
&\geq 0
\ee
for every $\ket{\psi} \in \C^d$. This can also be seen as a simple consequence of (the easy direction of) Bochner's theorem.
\end{proof}

\subsection{Proof of main result}

Before stating and proving our main result, let us fix some terminology. For an infinite matrix $\tau$ that is also a density operator on $\ell^2(\N)$, the associated \deff{maximally correlated state} $\Omega[\tau]$ on $\ell^2(\N)\otimes \ell^2(\N)$ is defined by
\bb
\Omega[\tau] \coloneqq \sum_{h,k=0}^\infty \tau_{hk} \ketbraa{h}{k}\otimes \ketbraa{h}{k}\, .
\label{maximally_correlated}
\ee
Maximally correlated states appear naturally in connecting coherence theory~\cite{Aberg2004, Aberg2006, Baumgraz2014, Winter2016, Chitambar-Hsieh-coherence, Regula2018, Fang2018, bound-coherence, GrandTour} (see also the review article~\cite{coherence-review}) with entanglement theory. When seen in this latter context, they are useful because they represent a particularly simple class of entangled states.

Before we proceed further, let us recall in passing that the R\'enyi entropy of a probability density defined on the interval $[-\pi,\pi]$ need not be finite; that is, it can be equal to $-\infty$. It is always bounded from above by $\log_2(2\pi)$, due to the non-negativity of relative entropy. Indeed, $h_\alpha(p) \leq \log_2 (2\pi)$ for every probability density $p$ defined on $[-\pi,\pi]$ and for all $\alpha\in (0,1)\cup(1,\infty)$, because $\log_2 (2\pi) - h_\alpha(p) = D_\alpha(p\Vert u) \geq 0$, where $u$ is the uniform probability density on $[-\pi,\pi]$. The same is true for the case $\alpha=1$, where we identify $h_1$ with the standard differential entropy~\eqref{differential_entropy}. However, as an example, if we take the probability density to be $p(x) = c |x|^{-\frac{1}{\alpha}}$ for $\alpha > 1$, where $c$ is a normalization factor, then the R\'enyi entropy $h_\alpha(p)$ diverges to $-\infty$. Note that the condition $h_{\alpha}(p) > -\infty$ is equivalent to the condition $\int_{-\pi}^{+\pi}d\phi\ p(\phi)^{\alpha} <\infty$. An analogous reasoning can be repeated for the case where $\alpha=1$. Here, $p(x) \propto |x|^{-1}\left(\log_2\frac{2\pi}{|x|}\right)^{-2}$ provides an example of a probability distribution for which $h(p) = -\infty$ (while $p$ itself is integrable). 

\begin{thm} \label{capacities_thm}
Let $p:[-\pi,+\pi]\to \R_+$ be a probability density function with the property that one of its R\'enyi entropies is finite for some $\alpha_0>1$, i.e.
\bb
\int_{-\pi}^{+\pi}d\phi\ p(\phi)^{\alpha_0} <\infty\, .
\label{finite_Renyi_p}
\ee
Then the two-way assisted quantum capacity, the unassisted quantum capacity, the private capacity, the secret-key capacity, and all of the corresponding strong converse capacities of the associated bosonic dephasing channel $\NN_p$ coincide, and are given by the expression
\bb
Q(\NN_p) &= Q^\dag(\NN_p) = P(\NN_p) = P^\dag(\NN_p) = \QQ(\NN_p) = \QQ^\dag(\NN_p) = P(\NN_p) = P_{\!\leftrightarrow}^\dag(\NN_p) \\
&= D(p\|u) = \log_2(2\pi) - h(p) \\
&= \log_2(2\pi) - \int_{-\pi}^{\pi} d\phi\ p(\phi) \log_2\frac{1}{p(\phi)}\, .
\ee
Here, $D(p\|u)$ denotes the Kullback--Leibler divergence between $p$ and the uniform probability density $u$ over $[-\pi,\pi]$, and $h(p)$ is the differential entropy of $p$.
\end{thm}

\begin{rem}
The condition on the R\'enyi entropy of $p$ is of a purely technical nature. We expect it to be obeyed in all cases of practical interest. For example, it holds true provided that $p$ is bounded on $[-\pi,\pi]$.
\end{rem}

\begin{rem}
By using H\"older's inequality, it can be easily verified that if~\eqref{finite_Renyi_p} holds for $\alpha_0\geq 1$ then it holds for all $\alpha$ such that $1\leq\alpha\leq \alpha_0$.
\end{rem}


\begin{proof}[Proof of Theorem~\ref{capacities_thm}]
The smallest of all eight quantities is the unassisted quantum capacity $Q(\NN_p)$, and the largest is $P_{\!\leftrightarrow}^\dag(\NN_p)$. Therefore, it suffices to prove that
\bb
Q(\NN_p) \geq D(p\|u)\, ,\qquad P_{\!\leftrightarrow}^\dag(\NN_p) \leq D(p\|u)\, .
\label{eq:bounds-on-caps}
\ee
Note that it is elementary to verify that
\bb
D(p\|u) = \int_{-\pi}^{+\pi} d\phi\ p(\phi) \log_2\frac{p(\phi)}{1/(2\pi)} = \log_2(2\pi) - \int_{-\pi}^{+\pi} d\phi\ p(\phi) \log_2\frac{1}{p(\phi)} = \log_2(2\pi) - h(p)\, ,
\ee
where the differential entropy $h(p)$ is defined by~\eqref{differential_entropy}.

To bound $Q(\NN_p)$ from below, we need an ansatz for a state $\ket{\Psi}_{AA'}$ to plug into~\eqref{LSD}. Letting $A$ and $A'$ be single-mode systems, we can consider the maximally entangled state $\ket{\Phi_d}_{AA'} \coloneqq \frac{1}{\sqrt{d}} \sum_{k=0}^{d-1} \ket{k}_A\ket{k}_{A'}$ locally supported on the subspace spanned by first $d$ Fock states $\ket{k}$ (see~\eqref{Fock}), where $k\in\{0,\ldots, d-1\}$. Let us also define the truncated matrix
\bb
T_p^{(d)} \coloneqq \Pi_d T_p \Pi_d = \sum_{h,k=0}^{d-1} (T_p)_{hk} \ketbraa{h}{k}\, ,
\label{truncated_Tp}
\ee
where
\bb
\Pi_d\coloneqq \sum_{k=0}^{d-1} \ketbra{k}\, .
\label{Pi_d}
\ee
Then note that
\bb
\omega_{p,d} \coloneqq&\ \big(I\otimes \NN_p\big)\big(\Phi_d\big) \\
=&\ \frac1d \sum_{h,k=0}^{d-1} \big(I\,\otimes\, \NN_p\big)\big(\ketbraa{hh}{kk}\big) \\
=&\ \frac1d \sum_{h,k=0}^{d-1} (T_p)_{hk} \ketbraa{hh}{kk} \\
=&\ \frac1d\, \Omega\big[ T_p^{(d)} \big]\, ,
\label{omega_p_d}
\ee
where $\Omega[\tau]$ is defined by~\eqref{maximally_correlated}. Consider that
\begin{align}
Q(\NN_p) &\textgeq{(i)} \limsup_{d\to\infty} \Icoh(A\rangle B)_{\big(I\,\otimes\, \NN_p^{A'\to B}\big)\big(\Phi_d^{AA'}\big)} \nonumber\\
&\texteq{(ii)} \limsup_{d\to\infty} \Icoh(A\rangle B)_{\omega_{p,d}} \nonumber\\
&\texteq{(iii)} \limsup_{d\to\infty} \left( \log_2 d - S\left( T_p^{(d)} \big/ d \right) \right) \nonumber\\
&\texteq{(iv)} \limsup_{d\to\infty} \left( \log_2 d + \frac1d \Tr T_p^{(d)}\! \left(-\log_2 d + \log_2 T_p^{(d)}\right) \right) \label{main_lower_bound_eq}
\\
&= \limsup_{d\to\infty} \frac1d \Tr T_p^{(d)} \log_2 T_p^{(d)} \nonumber\\
&\texteq{(v)} \int_{-\pi}^{\pi} \frac{d\phi}{2\pi}\ \left(2\pi p(\phi)\right) \log_2 \left(2\pi p(\phi)\right) \nonumber\\
&= \log_2(2\pi) - \int_{-\pi}^{\pi} d\phi\ p(\phi) \log_2 \frac{1}{p(\phi)}\, . \nonumber
\end{align}
Here: (i)~follows from the LSD theorem~\eqref{LSD}; in~(ii) we introduced the state $\omega_{p,d}$ defined by~\eqref{maximally_correlated}; (iii)~comes from~\eqref{Icoh}, due to the fact that $S(\Omega[\tau]) = S(\tau)$ on the one hand, and
\begin{align}
\Tr_A \omega_{p,d}^{AB} &= \Tr_A \big(I\otimes \NN_p^{A'\to B}\big)\big(\Phi_d^{AA'}\big) \nonumber\\
&= \frac1d \Tr_A\! \sum_{h,k=0}^{d-1} (T_p)_{hk} \ketbraa{h}{k}_A\otimes \ketbraa{h}{k}_B \nonumber\\
&= \frac1d \sum_{h,k=0}^{d-1} (T_p)_{hk} \delta_{hk} \ketbraa{h}{k}_B \\
&= \frac1d \sum_{k=0}^{d-1} \ketbra{k}_B \nonumber\\
&= \frac{\id_B}{d} \nonumber
\end{align}
and therefore $S\big( \Tr_A \big(I\,\otimes\, \NN_p^{A'\to B}\big)\big(\Phi_d^{AA'}\big)\big) = \log_2 d$ on the other; in~(iv) we simply substituted the definition~\eqref{von_Neumann} of von Neumann entropy; and finally in~(v) we employed Lemma~\ref{Serra-Capizzano_lemma} with the choice $a(\phi) = 2\pi p(\phi)$. This is possible due to our assumption that~\eqref{finite_Renyi_p} holds for some $\alpha>1$. Note that $F(x) = x\log_2 x$ satisfies $\left|F(x)\right| < x^\alpha$ for all $\alpha>1$ and for all sufficiently large $x\in \R_+$. This concludes the proof of the lower bound on $Q(\NN_p)$ in~\eqref{eq:bounds-on-caps}. We remark in passing that the Szeg\H{o} theorem has been applied before, although with an entirely different scope, in the context of quantum information theory~\cite{Lupo2010}.

We now establish the upper bound on $P_{\!\leftrightarrow}^\dag(\NN_p)$ in~\eqref{eq:bounds-on-caps}.
We claim that there is a sequence of LOCC protocols that can simulate $\NN_p$ using $\omega_{p,d}$ defined by~\eqref{omega_p_d} as a resource state and with error vanishing as $d\to\infty$. To see why this is the case, let $\rho$ be an arbitrary input state, and consider the $d$-dimensional teleportation protocol~\eqref{telep} on $\rho$ that uses $\omega_{p,d}$ as a resource. In formula, let us define
\bb
\NN_{p,d}^{A'\to B} (\rho_{A'})\coloneqq \T_{A'AB\to B}^{(d)} \left(\rho_{A'} \otimes \omega_{p,d}^{AB}\right) .
\label{simulation_NN_p_d}
\ee
We see that
\begin{align}
&\NN_{p,d}^{A'\to B} (\rho_{A'}) \nonumber\\
&\quad \texteq{(vi)} \sum_{x,z=0}^{d-1} X(x)_B Z(z)_B \Tr_{AA'}\! \left[ \rho_{A'} \otimes \omega_{p,d}^{AB}\, X(x)_{A'} Z(z)_{A'} \Phi_d^{AA'} Z(z)^\dag_{A'} X(x)^\dag_{A'} \right] Z(z)_B^\dag X(x)_B^\dag \nonumber\\
&\quad \texteq{(vii)} \sum_{x,z=0}^{d-1} \sum_{h,k=0}^{d-1} \frac1d\, (T_p)_{hk}\, X(x)_B Z(z)_B \nonumber\\
&\hspace{15.5ex} \Tr_{AA'}\! \left[ \rho_{A'} \otimes \ketbraa{hh}{kk}_{AB}\, X(x)_{A'} Z(z)_{A'} \Phi_d^{AA'} Z(z)^\dag_{A'} X(x)^\dag_{A'} \right] Z(z)_B^\dag X(x)_B^\dag \nonumber\\
&\quad \texteq{(viii)} \sum_{x,z=0}^{d-1} \sum_{h,k=0}^{d-1} \frac1d\, (T_p)_{hk}\, X(x)_B Z(z)_B \nonumber\\
&\hspace{15.5ex} \left(\Tr_{AA'}\! \left[ \rho_{A'} \otimes \ketbraa{h}{k}_{A}\, Z(z)_{A}^\intercal X(x)_{A}^\intercal \Phi_d^{AA'} X(x)_{A}^* Z(z)_{A}^* \right] \ketbraa{h}{k}_B \right) Z(z)_B^\dag X(x)_B^\dag \nonumber\\
&\quad \texteq{(ix)} \sum_{x,z=0}^{d-1} \sum_{h,k=0}^{d-1} \frac1d\, (T_p)_{hk}\, X(x)_B Z(z)_B \label{main_upper_bound_eq1} \\
&\hspace{15.5ex} \left( e^{\frac{2\pi i}{d} z(k-h)} \Tr_{AA'}\! \left[ \rho_{A'} \otimes \ketbraa{h\oplus x}{k\oplus x}_{A}\,  \Phi_d^{AA'} \right] \ketbraa{h}{k}_B \right) Z(z)_B^\dag X(x)_B^\dag \nonumber\\
&\quad \texteq{(x)} \sum_{x,z=0}^{d-1} \sum_{h,k=0}^{d-1} \frac{1}{d^2}\, (T_p)_{hk}\, X(x)_B Z(z)_B \left( e^{\frac{2\pi i}{d} z(k-h)} \rho_{h\oplus x,\, k\oplus x} \ketbraa{h}{k}_B \right) Z(z)_B^\dag X(x)_B^\dag \nonumber\\
&\quad \texteq{(xi)} \sum_{x,z=0}^{d-1} \sum_{h,k=0}^{d-1} \frac{1}{d^2}\, (T_p)_{hk}\, \rho_{h\oplus x,\, k\oplus x} \ketbraa{h\oplus x}{k\oplus x}_B \nonumber\\
&\quad = \sum_{x=0}^{d-1} \sum_{h,k=0}^{d-1} \frac1d\, (T_p)_{hk}\, \rho_{h\oplus x,\, k\oplus x} \ketbraa{h\oplus x}{k\oplus x}_B \nonumber\\
&\quad \texteq{(xii)} \sum_{h,k=0}^{d-1} \left(\frac1d  \sumno_{x=0}^{d-1} (T_p)_{h\oplus x,\, k\oplus x} \right) \rho_{hk} \ketbraa{h}{k}_B\, . \nonumber
\end{align}
In the above derivation, (vi)~follows from~\eqref{telep}, (vii)~from~\eqref{omega_p_d}, (viii)~from the formula
\bb
M\otimes \id \ket{\Phi_d} = \id \otimes M^\intercal \ket{\Phi_d}\, ,
\ee
valid for the maximally entangled state~\eqref{Phi_d} in any finite dimension, (ix)~and~(xi) from~\eqref{HW}, (x)~from the identity
\bb
\Tr \left[ M_A \otimes N_{A'} \Phi_d^{AA'} \right] = \frac1d \Tr \left[ MN^\intercal \right] ,
\ee
and finally (xii)~by a simple change of variable $h\oplus x \mapsto h$, $k\oplus x \mapsto k$, once one observes that
\bb
\sum_{x=0}^{d-1} (T_p)_{hk} \mapsto \sum_{x=0}^{d-1} (T_p)_{h\ominus x,\, k\ominus x} = \sum_{x'=0}^{d-1} (T_p)_{h\ominus (-x'),\, k\ominus (-x')} = \sum_{x'=0}^{d-1} (T_p)_{h\oplus x',\, k\oplus x'}\, ,
\ee
where $x'\coloneqq -x$.

The calculation in~\eqref{main_upper_bound_eq1} shows that
\bb
\braket{h|\NN_{p,d} (\rho) |k} = \left(\frac1d  \sumno_{x=0}^{d-1} (T_p)_{h\oplus x,\, k\oplus x}\right) \rho_{hk}\, .
\ee
We now want to argue that \emph{for fixed $h,k\in \N$} the above quantity converges to $(T_p)_{hk}\, \rho_{hk} = \braket{h|\NN_p(\rho)|k}$ as $d\to\infty$. To this end, note that if $d\geq h,k$ we have that $(T_p)_{h\oplus x,\, k\oplus x} = T_{hk}$ provided that either $x\leq \min\{d-1-h,d-1-k\} = \min\{d-h,d-k\}-1$ or $x\geq \max\{d-h,d-k\}$. Therefore, $(T_p)_{h\oplus x,\, k\oplus x} \neq T_{hk}$ for at most $|h-k|$ values of $x$, out of the $d$ possible ones. We can estimate the remainder terms pretty straightforwardly using the inequality $\left|(T_p)_{hk}\right|\leq 1$, valid for all $h,k\in \N$. Doing so yields 
\bb
\left|(T_p)_{hk} - \frac1d  \sum_{x=0}^{d-1} (T_p)_{h\oplus x,\, k\oplus x} \right| &\leq \left|(T_p)_{hk} - \frac{d-|h-k|}{d} (T_p)_{hk} \right| + \frac{|h-k|}{d} \\
&\leq \frac{2|h-k|}{d} \tends{}{d\to\infty} 0\, .
\label{main_upperbound_eq2}
\ee
Thus, for all fixed $h,k\in \N$,
\bb
\braket{h|\NN_{p,d} (\rho) |k} \tends{}{d\to\infty} \braket{h|\NN_{p} (\rho) |k}\qquad \forall\ \rho\, ,
\label{entrywise_convergence}
\ee
as claimed. We now argue that this implies the stronger fact that
\bb
\lim_{d\to\infty} \left\| \left(\left(\NN_{p,d} - \NN_p\right)_{A'\to B} \otimes I_E \right) (\rho_{A'E}) \right\|_1 = 0 \qquad \forall\ \rho_{A'E}\, ,
\label{strong_convergence}
\ee
where it is understood that $\rho_{A'E}$ is an arbitrary, but fixed state of a bipartite system $A'E$, with
the quantum system $E$ arbitrary. The above identity is usually expressed in words by saying that $\NN_{p,d}$ converges to $\NN_p$ in the \emph{topology of strong convergence}~\cite{Shirokov2008, Shirokov2018}. The arguments that allow to deduce~\eqref{strong_convergence} from~\eqref{entrywise_convergence} are standard:
\begin{enumerate}[(a)]
\item The linear span of the Fock states $\{\ket{k}\}_{k\in \N}$ is dense in $\HH_1$, and moreover the operators $\NN_{p,d} (\rho), \NN_{p} (\rho)$ are uniformly bounded in trace norm --- since they are states, they all have trace norm $1$. Hence, we see that~\eqref{entrywise_convergence} actually holds also when $\ket{h},\ket{k}$ are replaced by any two fixed vectors $\ket{\psi},\ket{\phi}\in \HH_1$. In formula,
\bb
\braket{\psi|\NN_{p,d} (\rho) |\phi} \tends{}{d\to\infty} \braket{\psi|\NN_{p} (\rho) |\phi}\qquad \forall\ \rho,\quad \forall\ \ket{\psi},\ket{\phi}\in \HH_1\, .
\label{weak_operator_convergence}
\ee
\item Therefore, by definition $\NN_{p,d} (\rho)$ converges to $\NN_{p} (\rho)$ in the \emph{weak operator topology}. Since the latter object is also a quantum state, an old result due to Davies~\cite[Lemma~4.3]{Davies1969}, which can also be seen as an elementary consequence of the `gentle measurement lemma'~\cite[Lemma~9]{VV1999} (see also~\cite[Lemmata~9.4.1 and~9.4.2]{MARK}), states that in fact there is trace norm convergence, i.e.
\bb
\lim_{d\to\infty} \left\| \left(\NN_{p,d} - \NN_p\right)(\rho) \right\|_1 = 0\qquad \forall\ \rho\, .
\label{strong_convergence_single_system}
\ee
\item The topology of strong convergence is stable under tensor products with the identity channel~\cite{Shirokov2008} (see also~\cite[Lemma~2]{Shirokov2018}). Therefore,~\eqref{strong_convergence_single_system} and~\eqref{strong_convergence} are in fact equivalent. Since we have proved the former, the latter also follows.
\end{enumerate}

We are now ready to prove that $P_{\!\leftrightarrow}^\dag(\NN_p)\leq D(p\|u)$. For a fixed positive integer $n\in \N_+$, consider a generic protocol as the one depicted in Figure~\ref{protocol_Q2_fig}, where the channel $\NN$ is now $\NN_p$. Since we are dealing with the secret-key-agreement capacity, the final state $\eta_n$ will approximate a private state $\gamma_n$ containing $\ceil{Rn}$ secret bits~\cite{private, Horodecki2009}. Here, $R$ is an achievable strong converse rate of secret-key agreement, i.e., it satisfies that $R<P_{\!\leftrightarrow}^\dag(\NN_p)$, where the right-hand side is defined by~\eqref{strong_converse_secret_key_agreement_capacity}. Call $\epsilon_n\coloneqq \frac12 \|\eta_n - \gamma_n\|_1$ the corresponding trace norm error, so that
\bb
\liminf_{n\to\infty} \epsilon_n <1\, .
\label{limsup_epsilon_n}
\ee
Imagine now to replace each instance of $\NN_p$ with its simulation $\NN_{p,d}$. This will yield at the output a state $\eta_{n,d}$, in general different from $\eta_n$; however, because of~\eqref{strong_convergence}, and since $n$ here is fixed, we have that the associated error $\delta_{n,d}$ vanishes as $d\to\infty$, i.e.
\bb
\delta_{n,d} \coloneqq \frac12 \left\|\eta_{n,d} - \eta_d \right\|_1 \tends{}{d\to\infty} 0\, .
\label{omega_n_d_approximation}
\ee

Now, after the above replacement the global protocol can be seen as an LOCC manipulation of $n$ copies of the state $\omega_{p,d}$ that is used to simulate $\NN_{p,d}$ as per~\eqref{simulation_NN_p_d}. By the triangle inequality, the trace distance between the final state $\eta_{n,d}$ and the private state $\gamma_n$ satisfies
\bb
\frac12 \left\| \eta_{n,d} - \gamma_n\right\|_1 \leq \frac12 \left\| \eta_{n,d} - \eta_n\right\|_1 + \frac12 \left\| \eta_{n} - \gamma_n\right\|_1 \leq \delta_{n,d} + \epsilon_n \, .
\ee
To apply the results of~\cite{MMMM}, we need to translate the above estimate into one that uses the fidelity instead of the trace distance. Such a translation can be made with the help of the Fuchs--van de Graaf inequalities~\cite{Fuchs1999}, here reported as~\eqref{Fuchs_vdG}. We obtain that $F(\eta_{n,d},\gamma_n)\geq \left(1- \delta_{n,d} - \epsilon_n\right)^2$. We can then use~\cite[Eq.~(5.37)]{MMMM} directly to deduce that
\bb
\left(1 - \delta_{n,d} - \epsilon_n\right)^2 \leq F(\eta_{n,d},\gamma_n) \leq 2^{- n\frac{\alpha-1}{\alpha} \left( R - \widetilde{E}_{R,\alpha}(\omega_{p,d}) \right)}
\label{converse_bound_eq1}
\ee
for all $1<\alpha\leq \alpha_0$, where
\bb
\widetilde{E}_{R,\alpha}(\rho_{AB}) \coloneqq \inf_{\sigma\in \SEP_{AB}} \widetilde{D}_\alpha(\rho\|\sigma)
\label{sandwiched_REE}
\ee
is the \deff{sandwiched $\boldsymbol{\alpha}$-R\'enyi relative entropy of entanglement}, and
\bb
\SEP_{AB} \coloneqq \mathrm{conv}\left\{ \ketbra{\psi}_A \otimes \ketbra{\phi}_B:\, \ket{\psi}_A\in \HH_A,\, \ket{\phi}_B\in \HH_B,\, \braket{\psi|\psi}=1=\braket{\phi|\phi} \right\}
\ee
is the set of \deff{separable states} over the bipartite quantum system $AB$. We can immediately recast~\eqref{converse_bound_eq1} as
\bb
R \leq \frac2n \frac{\alpha}{\alpha-1} \log_2\frac{1}{1-\delta_{n,d}-\epsilon_n} + \widetilde{E}_{R,\alpha}(\omega_{p,d})
\label{converse_bound_eq2}
\ee

Let us now estimate the quantity $\widetilde{E}_{R,\alpha}(\omega_{p,d})$. By taking as an ansatz for a separable state to be plugged into~\eqref{sandwiched_REE} simply $\Omega[\Pi_d/d]$ (see~\eqref{maximally_correlated} and~\eqref{Pi_d}), which is manifestly separable because $\Pi_d$ is diagonal, we conclude that
\bb
\widetilde{E}_{R,\alpha}(\omega_{p,d}) &\leq \widetilde{D}_\alpha\left( \omega_{p,d}\, \bigg\|\, \Omega\!\left[\frac{\Pi_d}{d}\right]\right) \\
&\texteq{(xiii)} \frac{1}{\alpha-1}\, \log_2 \Tr \omega_{p,d}^\alpha\, \Omega\!\left[\frac{\Pi_d}{d}\right]^{1-\alpha} \\
&\texteq{(xiv)} \frac{1}{\alpha-1}\, \log_2 \frac1d \Tr \left(T_{p}^{(d)}\right)^\alpha
\label{converse_bound_eq3}
\ee
Here, (xiii)~follows from~\eqref{commute_Petz--Renyi}, while in~(xiv) we simply recalled~\eqref{omega_p_d}.

Now, applying Lemma~\ref{Serra-Capizzano_lemma} once again with $a(\phi) = 2\pi p(\phi)$, we surmise that
\bb
\lim_{d\to\infty} \frac1d \Tr \left[\left(T_{p}^{(d)}\right)^\alpha\right] = \int_{-\pi}^{+\pi} \frac{d\phi}{2\pi}\ (2\pi p(\phi))^\alpha = (2\pi)^{\alpha-1} \int_{-\pi}^{+\pi} d\phi\ p(\phi)^\alpha\, .
\label{converse_bound_eq4}
\ee
Therefore, from~\eqref{converse_bound_eq3} we deduce that
\bb
\limsup_{d\to\infty} \widetilde{E}_{R,\alpha}(\omega_{p,d}) \leq \log_2(2\pi) + \frac{1}{\alpha - 1}\log_2 \int_{-\pi}^{+\pi} d\phi\ p(\phi)^\alpha = D_\alpha(p\|u)\, ,
\label{converse_bound_eq5}
\ee
where $D_\alpha$ is the $\alpha$-R\'enyi divergence defined by~\eqref{Renyi}. Due to both~\eqref{converse_bound_eq5} and~\eqref{omega_n_d_approximation}, taking the limit $d\to\infty$ in~\eqref{converse_bound_eq2} yields
\bb
R \leq \frac2n \frac{\alpha}{\alpha-1} \log_2\frac{1}{1-\epsilon_n} + D_\alpha(p\|u)\, .
\ee
We are now ready to take the limit $n\to\infty$. In light of~\eqref{limsup_epsilon_n}, we obtain that
\bb
R \leq \liminf_{n\to\infty} \left( \frac2n \frac{\alpha}{\alpha-1} \log_2\frac{1}{1-\epsilon_n} + D_\alpha(p\|u) \right) = D_\alpha(p\|u)\, .
\ee
The limit as $\alpha\to1^+$ can be computed via Lemma~\ref{technical_Renyi_lemma} (and in particular~\eqref{limit_1+_Renyi}), due to the condition~\eqref{finite_Renyi_p}, which can be rephrased as $D_{\alpha_0}(p\|u)<\infty$. It gives
\bb
R \leq \liminf_{\alpha\to 1^+} D_\alpha(p\|u) = D(p\|u)\, .
\ee
Since $R$ was an arbitrary achievable strong converse rate for secret-key agreement, we deduce that
\bb
P_{\!\leftrightarrow}^\dag(\NN_p) \leq D(p\|u)\, ,
\ee
completing the proof.
\end{proof}

\subsection{Extension to multimode channels}

We will now see how to extend our main result, Theorem~\ref{capacities_thm}, to the case of a multimode bosonic dephasing channel. An $m$-mode quantum system ($m\in \N_+$) is modelled mathematically by the Hilbert space $\HH_m = \HH_1^{\otimes m} = L^2(\R)^{\otimes m} = L^2(\R^m)$. The annihilation and creation operators $a_j,a_j^\dag$ ($j=1,\ldots,m$), defined by
\bb
a_1 \coloneqq a\otimes \id\otimes \ldots\otimes \id\, ,\quad\ldots ,\quad a_m \coloneqq \id\otimes\ldots \otimes \id\otimes a
\ee
in terms of the single-mode operators in~\eqref{a_adag}, satisfy the canonical commutation relations
\bb
\left[a_j,a_k\right] = 0 = \left[a_j^\dag, a_k^\dag\right] ,\qquad \left[a_j,a_k^\dag\right] = \delta_{jk} \id\, .
\ee
The multimode Fock states $\ket{\vb{k}}$, indexed by $\vb{k} = (k_1,\ldots,k_m)^\intercal\in \N^m$, are given by
\bb
\ket{\vb{k}} \coloneqq \ket{k_1} \otimes \cdots \otimes \ket{k_m}\, .
\ee

Now, for a probability density function $p$ on $[-\pi,\pi]^m$, the corresponding \deff{multimode bosonic dephasing channel} is the quantum channel $\NN_p^{(m)}:\T(\HH_m) \to \T(\HH_m)$ defined by
\bb
\NN_p^{(m)}(\rho) \coloneqq \int_{[-\pi,\pi]^m} d^m \vb{\upphi}\ p(\vb{\upphi})\, e^{-i \sum_j a^\dag_j a_j^{\phantom{\dag}} \phi_j} \rho\, e^{i\sum_j a^\dag_j a_j^{\phantom{\dag}} \phi_j} ,
\label{Np_multimode}
\ee
where $j=1,\ldots,m$, and $\vb{\upphi} = (\phi_1,\ldots,\phi_m)^\intercal$.

In perfect analogy with Theorem~\ref{capacities_thm}, we can now prove the following.

\begin{thm} \label{capacities_multimode_thm}
Let $p:[-\pi,+\pi]^m\to \R_+$ be a probability density function with the property that one of its R\'enyi entropies is finite for some $\alpha_0>1$, i.e.
\bb
\int_{[-\pi,\pi]^m} d^m \vb{\upphi}\ p(\vb{\upphi})^{\alpha_0} <\infty\, .
\label{finite_Renyi_p_multimode}
\ee
Then the two-way assisted quantum capacity, the unassisted quantum capacity, the private capacity, the secret-key capacity, and all of the corresponding strong converse rates of the associated multimode bosonic dephasing channel $\NN_p^{(m)}$ coincide, and are given by the expression
\bb
Q\big(\NN_p^{(m)}\big) &= Q^\dag\big(\NN_p^{(m)}\big) = P\big(\NN_p^{(m)}\big) = P^\dag\big(\NN_p^{(m)}\big) = \QQ\big(\NN_p^{(m)}\big) = \QQ^\dag\big(\NN_p^{(m)}\big) = P_{\!\leftrightarrow}\big(\NN_p^{(m)}\big) = P_{\!\leftrightarrow}^\dag\big(\NN_p^{(m)}\big) \\
&= D(p\|u) = m\log_2(2\pi) - h(p) \\
&= m\log_2(2\pi) - \int_{[-\pi,\pi]^m} d^m\vb{\upphi}\ p(\vb{\upphi}) \log_2\frac{1}{p(\vb{\upphi})}\, .
\ee
Here, $D(p\|u)$ denotes the Kullback--Leibler divergence between $p$ and the uniform probability distribution $u$ over $[-\pi,\pi]^m$, and $h(p)$ is the differential entropy of $p$.
\end{thm}

Rather unsurprisingly, one of the key technical tools that we need to prove the above generalization of Theorem~\ref{capacities_thm} is a multi-index version of the Szeg\H{o} theorem reported here as Lemma~\ref{Serra-Capizzano_lemma}. In fact, the original paper by Serra-Capizzano~\cite{SerraCapizzano2002} deals already with multi-indices, so we can borrow the following result directly from~\cite[Theorem~2]{SerraCapizzano2002} (cf.\ also the proof of Lemma~\ref{Serra-Capizzano_lemma}).

An \deff{multi-index infinite Toeplitz matrix} is an operator $T:\ell^2(\N^m) \to \ell^2(\N^m)$ with the property that its matrix entries $T_{\vb{h},\vb{k}}$ (where $\ket{h}=(h_1,\ldots,h_m)^\intercal\in \N^m$ is a multi-index) depend only on the difference $\vec{h}-\vec{k}$, in formula $T_{\vb{h},\vb{k}}= a_{\vb{h}-\vb{k}}$. The case of interest is when
\bb
a_{\vb{k}} = \int_{[-\pi,\pi]^m} d^m\vb{\upphi}\ a(\vb{\upphi})\, e^{-i\vb{k}\cdot\vb{\upphi}}\, ,
\ee
where $a:[-\pi,\pi]^m\to\R_+$ is a non-negative function, and $\vb{k}\cdot\vb{\upphi} \coloneqq \sum_{j=1}^m k_j\phi_j$. As in the setting of Szeg\H{o}'s theorem, one considers the truncations of $T$ defined for some $\vb{d} = (d_1,\ldots, d_m)^\intercal\in \N^m$ by
\bb
T^{(\vb{d})}\coloneqq \sum_{h_1,k_1=0}^{d_1-1}\ldots \sum_{h_m,k_m=0}^{d_m-1} T_{\vb{h},\vb{k}}\ketbraa{\vb{h}}{\vb{k}}\, .
\label{multiindex_truncation}
\ee
Note that $T^{(\vb{d})}$ is an operator on a space of dimension
\bb
D(\vb{d})\coloneqq \prod_{j=1}^m d_j\, .
\label{multiindex_total_dim}
\ee
The multi-index Szeg\H{o} theorem then reads
\bb
\lim_{\vb{d}\to\infty} \frac{1}{D\big(\vb{d}\big)} \Tr F\big( T^{(\vb{d})}\big) = \lim_{\vb{d}\to\infty} \frac{1}{D(\vb{d})} \sum_{j=1}^{D(\vb{d})} F\big(\lambda_j\big(T^{(\vb{d})}\big)\big) = \int_{[-\pi,\pi]^m} \frac{d\vb{\upphi}}{(2\pi)^m}\ F\big(a(\vb{\upphi})\big)\, ,
\label{Szego_multiindex}
\ee
where $F:\R\to\R$, and $\vb{d}\to\infty$ means that $\min_{j=1,\ldots,m} d_j \to\infty$. Conditions on $a$ and $F$ so that~\eqref{Szego_multiindex} holds are as follows.

\begin{lemma}[(Serra-Capizzano~\cite{SerraCapizzano2002}, multi-index case)] \label{Serra-Capizzano_multiindex_lemma}
If $a:[-\pi,\pi]^m\to \R_+$ is such that
\bb
\int_{[-\pi,\pi]^m} \frac{d^m\vb{\upphi}}{(2\pi)^m}\ a(\vb{\upphi})^\alpha < \infty
\ee
for some $\alpha \geq 1$, and moreover $F:\R_+ \to \R$ is continuous and satisfies
\bb
F(x) = O(x^\alpha) \qquad (x\to \infty)\, ,
\ee
then~\eqref{Szego_multiindex} holds.
\end{lemma}

The proof of Theorem~\ref{capacities_multimode_thm} follows very closely that of Theorem~\ref{capacities_thm}. Let us briefly summarize the main differences.

\begin{proof}[Proof of Theorem~\ref{capacities_multimode_thm}]
As an ansatz in the coherent information~\eqref{main_lower_bound_eq}, we use a multimode maximally entangled state, defined by
\bb
\ket{\Phi_{\vb{d}}}\coloneqq \frac{1}{D(\vb{d})} \sum_{k_1=0}^{d_1-1}\ldots \sum_{k_m=0}^{d_m-1} \ket{\vb{k}}_A\ket{\vb{k}}_{A'}\, ,
\ee
where $\vb{d}\in \N^m$ is fixed for now. Since
\bb
\omega_{p,\vb{d}} \coloneqq \left(I\otimes \NN_p^{(m)}\right)(\Phi_{\vb{d}}) = \frac{1}{D(\vb{d})} \, \Omega\big[T_p^{\vb{d}}\big]
\ee
is still a maximally correlated state, the derivation in~\eqref{main_lower_bound_eq} is unaffected, provided that one employs in~(v) Lemma~\ref{Serra-Capizzano_multiindex_lemma}.

As for the converse bound on the strong converse rate, one replaces~\eqref{simulation_NN_p_d} with
\bb
\big(\NN_{p,\vb{d}}^{(m)}\big)_{A'\to B} (\rho_{A'})\coloneqq \left(\bigotimes\nolimits_{j=1}^m \T_{A_j'A_jB_j\to B_j}^{(d_j)} \right) \left(\rho_{A'} \otimes \omega_{p,d}^{AB}\right) ,
\ee
where $A_j$ denotes the $j^{\text{th}}$ mode of $A$, and analogously for $A'$ and $B$. Then~\eqref{main_upper_bound_eq1} becomes
\bb
\big(\NN_{p,\vb{d}}^{(m)}\big)_{A'\to B} (\rho_{A'}) = \sum_{h_1,k_1=0}^{d_1-1}\ldots \sum_{h_m,k_m=0}^{d_m-1} \left(\frac{1}{D(\vb{d})} \sum_{x_1=0}^{d_1-1}\ldots \sum_{x_m=0}^{d_m-1} (T_p)_{\vb{h}\oplus\vb{x},\, \vb{k}\oplus \vb{x}} \right) \rho_{\vb{h},\vb{k}} \ketbraa{\vb{h}}{\vb{k}}_B\, .
\ee
We can write an inequality analogous to~\eqref{main_upperbound_eq2} as
\bb
&\left| (T_p)_{\vb{h},\vb{k}} - \frac{1}{D(\vb{d})}\, \sum_{x_1=0}^{d_1-1}\ldots \sum_{x_m=0}^{d_m-1} (T_p)_{\vb{h}\oplus\vb{x},\, \vb{k}\oplus \vb{x}} \right| \\ &\qquad \leq \left|(T_p)_{\vb{h},\vb{k}} - \frac{\prod_{j=1}^m (d_j-|h_j-k_j|)}{D(\vb{d})}\, (T_p)_{\vb{h},\vb{k}} \right| + \frac{D(\vb{d})-\prod_{j=1}^m (d_j-|h_j-k_j|)}{D(\vb{d})} \\
&\qquad \leq 2\frac{D(\vb{d})-\prod_{j=1}^m (d_j-|h_j-k_j|)}{D(\vb{d})} \tends{}{\vb{d}\to\infty} 0\, .
\ee
In the exact same way, one uses the above inequality to prove a generalized version of~\eqref{strong_convergence} as
\bb
\lim_{\vb{d}\to\infty} \left\| \left(\left(\NN_{p,\vb{d}}^{(m)} - \NN_p^{(m)}\right)_{A'\to B} \otimes I_E \right) (\rho_{A'E}) \right\|_1 = 0 \qquad \forall\ \rho_{A'E}\, .
\label{strong_convergence_multimode}
\ee
The combination of~\eqref{converse_bound_eq3} and~\eqref{converse_bound_eq4} now becomes
\bb
\widetilde{E}_{R,\alpha}(\omega_{p,\vb{d}}) \leq \frac{1}{\alpha-1} \log_2 \frac{1}{D(\vb{d})} \Tr \left[\left(T_p^{\vb{d}}\right)^\alpha\right] \tends{}{\vb{d}\to\infty} \frac{1}{\alpha-1}\log_2 (2\pi)^{m(\alpha-1)} \int_{[-\pi,\pi]^m} d\vb{\upphi}\, p(\vb{\upphi})^\alpha\, ,
\ee
so that we find, precisely as in~\eqref{converse_bound_eq5}, that 
\bb
\limsup_{\vb{d}\to\infty} \widetilde{E}_{R,\alpha}(\omega_{p,\vb{d}}) \leq m\log_2(2\pi) + \frac{1}{\alpha-1}\log_2 \int_{[-\pi,\pi]^m} d\vb{\upphi}\, p(\vb{\upphi})^\alpha = D_\alpha(p\|u)\, .
\ee
The rest of the proof is formally identical.
\end{proof}

\section{Examples}

\subsection{Wrapped normal distribution}

The most commonly studied~\cite{Arqand2020, Arqand2021} example of the bosonic dephasing channel is that which yields in~\eqref{Np_action_Tp} a matrix $T_p$ with entries
\bb
(T_{p_\gamma})_{hk} = e^{-\frac{\gamma}{2}(h-k)^2} ,
\ee 
where $\gamma>0$ is a parameter. The probability density function $p:[-\pi,+\pi]\to \R_+$ that gives rise to this matrix is a \deff{wrapped normal distribution}, that is, a normal distribution on $\R$ with variance $\gamma$ `wrapped' around the unit circle. In formula, this is given by
\bb
p_{\gamma}(\phi) = \frac{1}{\sqrt{2\pi\gamma}}\sum_{k=-\infty}^{+\infty} e^{-\frac{1}{2\gamma} (\phi+2\pi k)^2} .
\label{wrapped_Gaussian}
\ee
\begin{figure}[ht]
\begin{tikzpicture}[scale=1]
\begin{axis}[axis lines = left, xlabel = $\phi$,
xmin=-pi, xmax=pi, ymin=0, ymax=0.7,
xtick={-3.14159, -1.5708, 0, 1.5708, 3.14159}, xticklabels={$-\pi$,$-\pi/2$,$0$,$\pi/2$,$\pi$}, yticklabel style={/pgf/number format/fixed,/pgf/number format/precision=5},
scaled y ticks=false, legend style={at={(axis cs:(2.7,0.1)}, anchor=south, minimum height=0.5cm}]
\addplot[line width=1pt, solid,color=black] table[x=a, y=b, col sep=comma]{distributions.csv};
\addlegendentry{$p_\gamma(\phi)$};
\addplot[line width=1pt, solid,color=Blues5seq5] table[x=a, y=c, col sep=comma]{distributions.csv};
\addlegendentry{$p_\lambda(\phi)$};
\addplot[line width=1pt, solid,color=Reds5seq4] table[x=a, y=d, col sep=comma]{distributions.csv};
\addlegendentry{$p_\kappa(\phi)$};
\end{axis}
\end{tikzpicture}
\caption{The probability density functions of the wrapped normal~\eqref{wrapped_Gaussian}, von Mises~\eqref{von_Mises}, and wrapped Cauchy distributions~\eqref{eq:wrapped-cauchy-cap}, plotted as a function of $\phi\in [-\pi,\pi]$ for the case where $\gamma=\lambda=\kappa=0.5$.}
\label{distributions_fig}
\end{figure}
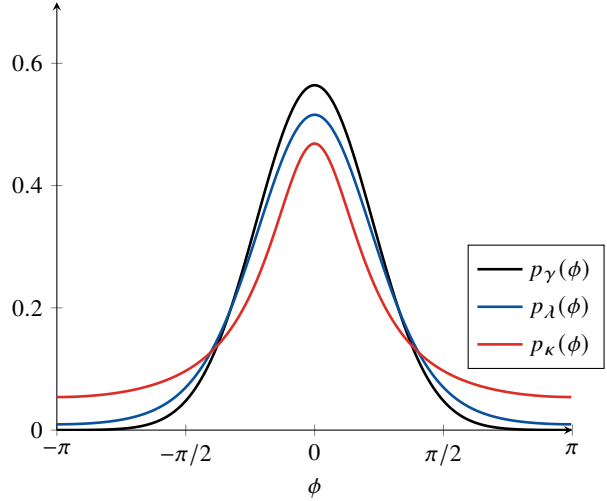
Its entropy can be expressed as~\cite[Chapter~3, \S~3.3]{PAPADIMITRIOU}
\bb
h(p_\gamma) = \frac{1}{\ln 2}\left( - \ln \left(\frac{\varphi(e^{-\gamma})}{2\pi}\right) + 2 \sum_{k=1}^\infty \frac{(-1)^k e^{-\frac{\gamma}{2}(k^2+k)}}{k\left(1-e^{-k\gamma}\right)} \right)\, ,
\ee
where
\bb
\varphi(q) \coloneqq \prod_{k=1}^\infty \left(1-q^k\right)
\ee
is the Euler function. Therefore, the capacities of the channel $\NN_{p_\gamma}$ are given by
\bb
Q(\NN_{p_\gamma}) &= Q^\dag(\NN_{p_\gamma}) = P(\NN_{p_\gamma}) = P^\dag(\NN_{p_\gamma}) = \QQ(\NN_{p_\gamma}) = \QQ^\dag(\NN_{p_\gamma}) = P_\leftrightarrow(\NN_{p_\gamma}) = P_\leftrightarrow^\dag(\NN_{p_\gamma}) \\
&= D(p_\gamma \| u) = \log_2 \varphi( e^{-\gamma} ) + \frac{2}{\ln 2} \sum_{k=1}^\infty \frac{(-1)^{k-1} e^{-\frac{\gamma}{2}(k^2+k)}}{k\left(1-e^{-k\gamma}\right)}\, .
\label{capacities_wrapped_Gaussian}
\ee

It is instructive to obtain asymptotic expansions of the above expressions in the limits $\gamma\ll 1$ (small dephasing) and $\gamma\gg 1$ (large dephasing).
\begin{itemize}
\item \emph{Small dephasing.} When $\gamma \to 0^+$, the channel $\NN_{p_\gamma}$ approaches the identity over an infinite-dimensional Hilbert space. Therefore, it is intuitive to expect that its capacities will diverge. To determine its asymptotic behavior, it suffices to note that in this limit the entropy of the wrapped normal distribution, which is very concentrated around $0$, is well approximated by that of the corresponding normal variable on the whole $\R$, i.e., $\frac12 \log_2 (2\pi e \gamma)$. Thus
\bb
Q(\NN_{p_\gamma}) &= Q^\dag(\NN_{p_\gamma}) = P(\NN_{p_\gamma}) = P^\dag(\NN_{p_\gamma}) = \QQ(\NN_{p_\gamma}) = \QQ^\dag(\NN_{p_\gamma}) = P_\leftrightarrow(\NN_{p_\gamma}) = P_\leftrightarrow^\dag(\NN_{p_\gamma}) \\
&\approx \frac12 \log_2\frac{2\pi}{e\gamma}\, .
\ee
In practice, already for $\gamma\lesssim 1$ the above estimate is within about $1\%$ of the actual value.

\item \emph{Large dephasing.} A straightforward computation using the series representation
\bb
-\ln \varphi(q) = \sum_{k=1}^\infty \frac{1}{k}\frac{q^k}{1-q^k}
\ee
yields the expansion
\bb
\ln \varphi(q) + 2 \sum_{k=1}^\infty \frac{(-1)^{k-1} q^{-\frac{1}{2}(k^2+k)}}{k\left(1-q^k\right)} &= q+\frac{q^2}{2}-\frac{q^3}{3}+\frac{q^4}{4}-\frac{q^5}{5}+\frac{2 q^6}{3}+O\Big(q^7\Big) \\
&= 2q - \ln(1+q) + O(q^6)\, .
\ee
This can be plugged into~\eqref{capacities_wrapped_Gaussian} to give
\bb
Q(\NN_{p_\gamma}) &= Q^\dag(\NN_{p_\gamma}) = P(\NN_{p_\gamma}) = P^\dag(\NN_{p_\gamma}) = \QQ(\NN_{p_\gamma}) = \QQ^\dag(\NN_{p_\gamma}) = P_\leftrightarrow(\NN_{p_\gamma}) = P_\leftrightarrow^\dag(\NN_{p_\gamma}) \\
&= \frac{2}{\ln 2}\, e^{-\gamma} - \log_2 \left(1+e^{-\gamma} \right) + O\left( e^{-6\gamma}\right) \\
&= \frac{e^{-\gamma}}{\ln 2} + O\left(e^{-2\gamma}\right) .
\ee
\end{itemize}

Incidentally, the combination of these two regimes yields an excellent approximation of the capacities across the whole range of $\gamma>0$. Namely,
\bb
Q(\NN_{p_\gamma}) &= Q^\dag(\NN_{p_\gamma}) = P(\NN_{p_\gamma}) = P^\dag(\NN_{p_\gamma}) = \QQ(\NN_{p_\gamma}) = \QQ^\dag(\NN_{p_\gamma}) = P_\leftrightarrow(\NN_{p_\gamma}) = P_\leftrightarrow^\dag(\NN_{p_\gamma}) \\
&\approx \max\left\{ \frac12 \log_2\frac{2\pi}{e\gamma},\ \frac{2}{\ln 2}\, e^{-\gamma} - \log_2 \left(1+e^{-\gamma} \right) \right\} .
\ee
The maximum absolute difference between the left-hand side and the right-hand side for $\gamma>0$ is less than $4\times 10^{-3}$.

\subsection{Von Mises distribution}

The \deff{von Mises distribution} on $[-\pi,+\pi]$ is defined by
\bb
p_\lambda(\phi) = \frac{e^{\frac1\lambda \cos(\phi)}}{2\pi\, I_0(1/\lambda)}\, ,
\label{von_Mises}
\ee
where $I_n$ denotes a modified Bessel function of the first kind. Here, $\lambda>0$ is a parameter that plays a role analogous to that $\gamma>0$ played in the case of the wrapped normal. The matrix $T_{p_\lambda}$ obtained in~\eqref{Np_action_Tp} for $p=p_\lambda$ is given by
\bb
(T_{p_\lambda})_{hk} = \frac{I_{|h-k|}(1/\lambda)}{I_0(1/\lambda)}\, .
\ee
The differential entropy of $p_\lambda$ can be calculated analytically, yielding~\cite[Chapter~3, Section~3.3]{PAPADIMITRIOU}
\bb
h(p_\lambda) = \log_2 (2\pi\, I_0(1/\lambda)) - \frac{1}{\ln 2} \frac{I_1(1/\lambda)}{\lambda\, I_0(1/\lambda)}\, .
\ee
Therefore, the capacities of the corresponding bosonic dephasing channel are given by
\bb
Q(\NN_{p_\lambda}) &= Q^\dag(\NN_{p_\lambda}) = P(\NN_{p_\lambda}) = P^\dag(\NN_{p_\lambda}) = \QQ(\NN_{p_\lambda}) = \QQ^\dag(\NN_{p_\lambda}) = P_\leftrightarrow(\NN_{p_\lambda}) = P_\leftrightarrow^\dag(\NN_{p_\lambda}) \\
&= \frac{1}{\ln 2} \frac{I_1(1/\lambda)}{\lambda\, I_0(1/\lambda)} - \log_2 I_0(1/\lambda)\, .
\ee

\subsection{Wrapped Cauchy distribution}

Our final example of a probability distribution on the circle, and of the bosonic dephasing channel associated to it, is defined similarly to the wrapped normal distribution, but this time starting from the Cauchy probability density function. Namely, for some parameter $\kappa>0$ we set
\bb
p_\kappa(\phi) \coloneqq \sum_{k=-\infty}^{+\infty} \frac{\sqrt{\kappa}}{\pi\left(\kappa+\left(\phi+2\pi k\right)^2\right)} = \frac{1}{2\pi} \frac{\sinh(\sqrt{\kappa})}{\cosh(\sqrt{\kappa}) -\cos\phi}.
\label{wrapped_Cauchy}
\ee
For a proof of the second identity, see~\cite[p.~51]{MARDIA}. The matrix $T_{p_\kappa}$ obtained in~\eqref{Np_action_Tp} for $p=p_\kappa$ is given by
\bb
(T_{p_\kappa})_{hk} = e^{-\sqrt{\kappa}|h-k|}\, .
\ee
The differential entropy of $p_\kappa$ is equal to $\log_2\big(2\pi\big(1-e^{-2\sqrt{\kappa}}\big)\big)$~\cite[Chapter~3, \S~3.3]{PAPADIMITRIOU}, implying that the various capacities of the corresponding bosonic dephasing channel $\NN_{p_\kappa}$ are equal to
\bb
\CC(\NN_{p_\kappa}) = \log_2\!\left(\frac{1}{1-e^{-2\sqrt{\kappa}}}\right) \, .
\label{eq:wrapped-cauchy-cap}
\ee


\end{document}